\documentclass[11pt, a4paper]{article}
\usepackage[top=2.5cm, bottom=3.2cm, left=2.5cm, right=2.5cm]{geometry}
\usepackage{setspace}
\usepackage{amsfonts}
\usepackage{amssymb,amsthm,amsmath}
\usepackage{color}
\usepackage{enumerate}
\usepackage{authblk}
\usepackage[none]{hyphenat}
\usepackage{appendix}
\usepackage{hyperref}
\allowdisplaybreaks

\numberwithin{equation}{section}

\newtheorem{theorem}{Theorem}[section]
\newtheorem{lemma}[theorem]{Lemma}
\newtheorem{proposition}[theorem]{Proposition}

\newtheorem{remark}{Remark}

\newcommand{\bea}{\begin{eqnarray}}
\newcommand{\eea}{\end{eqnarray}}
\newcommand{\bean}{\begin{eqnarray*}}
\newcommand{\eean}{\end{eqnarray*}}

\newcommand{\al}{\alpha}

\renewcommand{\hat}{\widehat}

\newcommand{\sqrtx}{\sqrt{(b-x)(x-a)}}

\title{Hankel Determinant for a Perturbed Laguerre Weight with Pole Singularities and Generalized Painlev\'{e} III$'$ Equation}
\author[1,*]{Shulin Lyu}
\author[1,$\dag$]{Yuanfei Lyu}  
\affil[1]{School of Mathematics and Statistics, Qilu University of Technology (Shandong Academy of Sciences), Jinan 250353, China}
\affil[*]{\texttt{lvshulin1989@163.com}, Corresponding author}
\affil[$\dag$]{\texttt{lvyuanfei5720@163.com}}
\date{\today}

\begin{document}
\maketitle

\begin{abstract}
We study the Hankel determinant for the weight $x^{\al}{\rm exp}(-x-t_1/x-t_2/x^2), x\in[0,+\infty)$, with $\alpha>-1,~t_1\in\mathbb{R}\setminus\{0\}, ~t_2>0.$
Compared with the weight $x^{\alpha}{\rm e}^{-x-t_1/x}$ studied in prior work (where $\alpha,t_1>0$), the range of $\alpha$ in our work is extended and the parameter $t_2$ introduces a ``stronger" zero at the origin. This leads to more varied behavior of the Hankel determinant, and the interplay between $t_1$ and $t_2$ introduces uncertainty and complexity into the analysis. By using a pair of ladder operators satisfied by the associated monic orthogonal polynomials and three compatibility conditions, we show that the recurrence coefficients are expressed in terms of four auxiliary quantities which satisfy a system of difference equations that can be iterated. We also establish two coupled second order partial differential equations (PDEs) satisfied by two of the auxiliary quantities, which are reduced to a Painlev\'{e} III$^\prime$ equation when $t_2\rightarrow0^+$. Moreover, the logarithmic derivative of the Hankel determinant is shown to satisfy a second order six degree PDE which is reduced to the $\sigma$-form of the Painlev\'{e} III$^{\prime}$ equation when $t_2\rightarrow0^+$. Under suitable double scaling, we obtain the limiting forms of the above PDEs and deduce the equilibrium density of the eigenvalues for the unitary ensemble. We extend our analysis to the Hankel determinant for $x^{\alpha}\exp(-x-\sum_{k=1}^m t_k/x^k)$ with $m=3$. For general $m$, we outline a derivation that leads, at least in principle, to the PDE satisfied by the logarithmic derivative of the Hankel determinant.
\end{abstract}

\noindent
$\mathbf{Keywords}$: Laguerre unitary ensembles; Hankel determinant; Orthogonal polynomials;\\
Painlev\'{e} equations

\noindent
$\mathbf{Mathematics\:\: Subject\:\: Classification\:\: 2020}$: 15B52; 33E17; 42C05

\sloppy{
\section{Introduction}

We consider the $n$-dimensional singularly perturbed Laguerre unitary ensemble of $n\times n$ Hermitian matrices whose eigenvalues have the following joint probability density function
\begin{align}\label{jpdf}
p(x_1,\dots,x_n)=\frac{1}{n!D_n}\prod_{1\leq i<j\leq n}(x_i-x_j)^2\prod_{k=1}^n w(x_k;t_1,t_2),
\end{align}
where the weight function $w(x;t_1,t_2)$ is obtained through a deformation of the Laguerre weight, i.e.
\begin{align}\label{w}
w(x; t_1,t_2)=x^{\al}{\rm exp}\left(-x-\frac{t_1}{x}-\frac{t_2}{x^2}\right),\quad x\in[0,+\infty), ~\alpha>-1,~t_1\in\mathbb{R}\setminus\{0\}, ~t_2>0.
\end{align}
 The reason why we assume $t_1\neq0$ is that the approach presented in this paper is inapplicable for $t_1=0$.
The term $n!D_n$ in \eqref{jpdf} is known as the partition function and $D_n$ is given by
\begin{align}
D_n(t_1,t_2)=&\frac{1}{n!}\int_{[0,+\infty)^n}\prod\limits_{1\leq i<j\leq n}(x_i-x_j)^2\prod_{k=1}^n w(x_k;t_1,t_2)dx_1\cdots dx_n\nonumber\\
=&\det\left(\int_0^{+\infty}x^{i+j}w(x;t_1,t_2)dx\right)_{i,j=0}^{n-1},\label{HD}
\end{align}
where the second equality is due to Heine's formula \cite{Szego1939}. This paper is concerned with $D_n(t_1,t_2)$.

The weight $w(x;t_1,0)$, i.e. $x^{\alpha}{\rm e}^{-x-t_1/x}$ with $t_1>0$, arises from integrable quantum field theories at finite temperature and was studied in \cite{ChenIts2010} under the assumption that $\alpha>0$.
By using the ladder operators and the Lax pair of the Riemann-Hilbert problem satisfied by the associated monic orthogonal polynomials, $D_n(t_1,0)$ was expressed as an integral in terms of the solution of a Painlev\'{e} III$'$ equation. Compared with the above case, the range of $\alpha$ is extended in $w(x;t_1,t_2)$ and the factor ${\rm e}^{-t_2/x^2}$ introduces a ``stronger" zero at the origin, which leads to more varied behavior of $D_n$. In addition, the interplay between $t_1$ and $t_2$ introduces uncertainty and complexity into the analysis.

The Hankel determinant generated by the weight function $y^{\alpha}{\rm e}^{-n(y+t^2/y^2)}$ was studied in \cite{ACM} (see (1.51) therein with $k=2$). It multiplied by $n!$ denotes the partition function of the associated unitary ensemble and  also the moment generating function of the charge relaxation resistance \cite{ACM}. For other motivations to study this Hankel determinant, refer to Section 1.1 of \cite{ACM}. Under a suitable double scaling, the eigenvalue correlation kernel of this ensemble was characterized by a hierarchy of higher order Painlev\'{e} III equation (see  Theorem 1.11 of \cite{ACM}). By a change of variable $y=x/n$, this Hankel determinant is transformed into the one for the weight $ w(x;0,n^3t^2)$, up to a multiplicative constant.

It is well-known that the Hankel determinant admits the following representation \cite{Ismail2005}:
\begin{align}\label{Dnhn}
D_{n}(t_1,t_2)=\prod_{j=0}^{n-1}h_{j}(t_1,t_2),
\end{align}
where $h_j$ is the square of the $L^2$-norm of the $j$th-degree monic polynomial $P_j(x;t_1,t_2)$ orthogonal with respect to $w(x;t_1,t_2)$, namely,
\bea\label{or}
\int_{0}^{+\infty}P_{j}(x;t_1,t_2)P_{k}(x;t_1,t_2)w(x;t_1,t_2)dx=h_{j}(t_1,t_2)\delta_{jk},\qquad j,k=0,1,\ldots.
\eea
Here $\delta_{jk}=1$ for $j=k$ and 0 otherwise, and
\bea\label{Pndef}
P_{j}(x;t_1,t_2)=x^{j}+p(j,t_1,t_2)x^{j-1}+\cdots+P_{j}(0;t_1,t_2)
\eea
for $j\geq1$ and $P_0(x;t_1,t_2):=1$.

According to \eqref{or}-\eqref{Pndef}, one shows that $\{P_j(x;t_1,t_2)\}$ satisfy a three-term recurrence relation
\begin{align}\label{Pnrecu}
xP_{j}(x;t_1,t_2)=P_{j+1}(x;t_1,t_2)+\alpha_{j}(t_1,t_2)P_{j}(x;t_1,t_2)+\beta_{j}(t_1,t_2)P_{j-1}(x;t_1,t_2),\quad j=0,1,\ldots,
\end{align}
subject to the initial condition $\beta_{0}(t_1,t_2)P_{-1}(x;t_1,t_2):=0.$
Inserting \eqref{Pndef} into \eqref{Pnrecu} and equating the coefficients of $x^j$, we find
\bea\label{a13}
\alpha_{j}(t_1,t_2)=p(j,t_1,t_2)-p(j+1,t_1,t_2),\qquad\qquad j\geq 0,
\eea
with the initial value $p(0,t_1,t_2):=0$, which immediately gives us
\bea\label{a01}
\sum_{j=0}^{n-1}\alpha_{j}(t_1,t_2)=-p(n,t_1,t_2).
\eea
Multiplying both sides of \eqref{Pnrecu} by $P_{j-1}(x;t_1,t_2)w(x;t_1,t_2)$ and integrating it from 0 to $+\infty$, in view of \eqref{or}, we get
\bea\label{bt-1}
\beta_{j}(t_1,t_2)h_{j-1}(t_1,t_2)=\int_0^{+\infty}x P_{j-1}(x;t_1,t_2)P_j(x;t_1,t_2)w(x;t_1,t_2)dx=h_j(t_1,t_2),
\eea
where the second equality is obtained via a replacement of $xP_{j-1}$ by using \eqref{Pnrecu} with $j-1$ in place of $j$. We readily see from \eqref{bt-1} that
\bea\label{a12}
\beta_{j}(t_1,t_2)=\frac{h_{j}(t_1,t_2)}{h_{j-1}(t_1,t_2)},
\eea
According to the recurrence relation \eqref{Pnrecu}, one derives the Christoffel-Darboux formula
\begin{align*} \sum_{j=0}^{n-1}\frac{P_j(x;t_1,t_2)P_j(y;t_1,t_2)}{h_j(t_1,t_2)}= \frac{P_n(x;t_1,t_2)P_{n-1}(y;t_1,t_2)-P_n(y;t_1,t_2)P_{n-1}(x;t_1,t_2)}{h_{n-1}(t_1,t_2)(x-y)}.
\end{align*}

By using this formula, the recurrence relation \eqref{Pnrecu} and the definition of $P_n(z):=P_n(z;t_1,t_2)$ given by \eqref{or}-\eqref{Pndef}, the following pair of lowering and raising operators were deduced in \cite{MinFang24}:
\begin{align}
\left(\frac{d}{dz}+B_{n}(z)\right)P_{n}(z)=&\beta_{n}A_{n}(z)P_{n-1}(z),\label{lo}\\
\left(\frac{d}{dz}-B_{n}(z)-v'(z)\right)P_{n-1}(z)=&-A_{n-1}(z)P_{n}(z),\label{ro}
\end{align}
for $n\geq0$, where $v(z)=v(z;t_1,t_2):=-\ln\:w(z;t_1,t_2)$ and
\begin{gather}
A_{n}(z):=\frac{1}{z}\cdot\frac{1}{h_{n}}\int_{0}^{+\infty}\frac{z v'(z)-xv'(x)}{z-x}P_{n}^{2}(x)w(x;t_1,t_2)dx,\qquad n\geq0,\label{defAn}\\
B_{n}(z):=\frac{1}{z}\left(\frac{1}{h_{n-1}}\int_{0}^{+\infty}\frac{z v'(z)-xv'(x)}{z-x}P_{n}(x)P_{n-1}(x)w(x;t_1,t_2)dx-n\right),\qquad n\geq1,\label{defBn}
\end{gather}
with the initial conditions $A_{-1}(z)=B_0(z):=0$. They can also be derived by employing the Riemann-Hilbert problem for $P_n(z)$ \cite{LyuLyu25}.
Moreover, it was shown in \cite{MinFang24} and \cite{LyuLyu25} that $A_n$ and $B_n$ satisfy the following identities
\begin{gather}
B_{n+1}(z)+B_{n}(z)=\left(z-\alpha_{n}\right)A_{n}(z)-v'(z),\tag{$S_1$}\\
1+\left(z-\alpha_{n}\right)\left(B_{n+1}(z)-B_{n}(z)\right)=\beta_{n+1}A_{n+1}(z)-\beta_{n}A_{n-1}(z),\tag{$S_2$}\\
\left(B_{n}(z)+v'(z)\right)B_{n}(z)+\sum_{j=0}^{n-1}A_{j}(z)=\beta_{n}A_{n}(z)A_{n-1}(z).\tag{$S_2'$}
\end{gather}
Actually, it was proved in \cite{MinFang24} and \cite{LyuLyu25} that these three compatibility conditions and the above ladder operators are valid for a more general Laguerre-type weight function of the form $x^{\alpha}w_0(x),x\in[0,+\infty)$ with $\alpha>-1$, where $w_0(x)>0$ was assumed to be continuously differentiable and $\lim\limits_{x\rightarrow+\infty}\pi_n(x)x^{\alpha}w_0(x)=0$ for an arbitrary polynomial $\pi_n(x)$ of order $n\geq0$. For more discussions on ladder operators, we refer the readers to \cite{ChenIsmail1997, ChenLyu23,Magnus1995, MuLyu23, WVA}.

For $w(x;t_1,t_2)$ given by \eqref{w}, we get the expressions for $A_n$ and $B_n$ from \eqref{defAn}-\eqref{defBn}, with four auxiliary quantities introduced. Substituting them into $(S_1)$ and $(S_2')$, we express $\alpha_n,\beta_n$ and $p(n,t_1,t_2)$ in terms of the auxiliary quantities which are shown to satisfy a system of difference equations that can be iterated in $n$. In addition, by differentiating the orthogonality relations \eqref{or} for $(j,k)=(n,n)$ and $(n,n-1)$ with respect to $t_1$ and $t_2$, we find differential relations between $\{h_n, p(n,t_1,t_2)\}$  and the auxiliary quantities.
Combining all the equalities obtained together, we deduce Riccati-like equations for the auxiliary quantities, from which a system of second order PDEs are derived. And moreover, we arrive at the second order PDE satisfied by the logarithmic derivative of $D_n(t_1,t_2)$. When $t_2\rightarrow0^+$, these PDEs are reduced to the Painlev\'{e} III$^\prime$ equation and its $\sigma$-form.

The above derivation strategy is called the ladder operator approach which is elementary, straightforward and powerful. It has been applied to study various finite-dimensional problems in unitary ensembles, including the classical orthogonal polynomials, partition functions, gap probabilities, the moment generating function of linear statistics etc. For example, in \cite{ChenIsmail2004}, the ladder operator approach was used to derive the recurrence coefficients, the $L^2$-norm and the explicit representations for the monic Jacobi polynomials. For problems involving one variable, a Painlev\'{e} equation and its $\sigma$-form are usually established.  The Hankel determinants generated by the Laguerre and Jacobi weight multiplied by the factor $(x+t)^{\lambda}$ with $t>0$ are connected with multiple-input multiple-output wireless communication system \cite{ChenMcKay2012}, and the logarithmic derivatives were shown to be directly related to the $\sigma$-form of the Painlev\'{e} V and VI equations, respectively. The gap probabilities of Gaussian and Jacobi unitary ensembles on $(-a,a),a>0$, were studied in \cite{LyuChenFan18,MinChen17}, and under a suitable double scaling the celebrated JMMS equation was recovered for both cases.
The ladder operator approach was also used to study the Hankel determinants for singularly perturbed Gaussian and Jacobi weights which are even \cite{MinChen20,MinChen21, MinLyuChen18}. For problems involving two variables, the logarithmic derivative of the Hankel determinant was shown to satisfy a second order PDE which is regarded as a two-variable generalization of a Painlev\'{e} equation; see \cite{BasorChenZhang, ChenHaqMcKay, LyuGriffinChen, MinChen19}. For problems involving more than two variables, see recent works \cite{MuLyu23, MuLyu24}.

The orthogonal polynomials and Hankel determinant associated with the Gaussian weight multiplied by a factor with several jump discontinuities were addressed in \cite{WuXu21}. When the dimension is finite and with the aid of the Lax pair, they were characterized by a coupled Painlev\'{e} IV system. Under a double scaling and based on the Deift-Zhou nonlinear steepest descent analysis on the Riemann-Hilbert problem (we call it RH method below), the asymptotics for them were obtained in terms of solutions of a coupled Painlev\'{e} II system. This Hankel determinant was investigated in  \cite{ChenLyu23} by employing the ladder operator approach and a system of second order PDEs were derived. Moreover, with the help of the Lax pair, direct relationships were found between the solutions of these PDEs and the coupled Painlev\'{e} IV system from \cite{WuXu21}. The above research strategy employed in \cite{ChenLyu23} was also applied to study the Hankel determinant for the deformed Laguerre weight with several jump discontinuities \cite{LyuChenXu23}.

The RH method was widely used to study the asymptotic behavior of unitary ensembles, including the partition functions \cite{CG21}, orthogonal polynomials \cite{WangFan22}, gap probabilities \cite{XuDai19},  the correlation kernel \cite{XuDaiZhao14} etc, and their asymptotics were expressed in terms of solutions of Painlev\'{e} equations. For problems involving two or more variables, a coupled Painlev\'{e} system was derived. We present here the main results of some literature which is concerned with singularly perturbed unitary ensembles with pole singularities in the potential. In \cite{ACM}, by introducing a model Riemann-Hilbert problem, the double scaling asymptotics for the partition functions of the singularly perturbed Gaussian and Laguerre unitary ensembles with a pole of order $2k$ and $k$ ($k\geq1$) respectively at $x=0$ were given in terms of hierarchies of Painlev\'{e} III equations. Moreover, for more general weight function $\exp(V(x)+(t/x)^k)$ with $V(x)>0$ being real analytic, the eigenvalue correlation kernel of the corresponding unitary ensemble was also described by a Painlev\'{e} III hierarchy. In \cite{BMM}, under a double scaling limit, the asymptotics of the partition function for a singularly perturbed Gaussian unitary ensemble with pole singularities of first and second order at $x=0$ in the potential was expressed in terms of the solution of a Painlev\'{e} III equation. In \cite{DaiXuZhang18}, the gap probability at the hard edge for the singularly perturbed unitary ensembles with pole singularities of order $k$ at the origin in the potential was expressed as an integral of the solution of a coupled Painlev\'{e} III system. In \cite{DaiXuZhang19}, the partition function and  correlation kernel for a singularly perturbed Gaussian unitary ensemble with pole singularities at $x=\lambda$ of orders $k=1,2,\dots,2m$ in the potential were characterized by the generalization of a Painlev\'{e} XXXIV equation at the soft edge.

This paper is built up as follows. In the next section, we express the recurrence coefficients in terms of four auxiliary quantities which satisfy a system of difference equations that can be iterated. Section 3 is devoted to the derivation of Toda-like equations for the recurrence coefficients and PDEs satisfied by the auxiliary quantities and the logarithmic derivative of the Hankel determinant. Under a double scaling, the limiting PDEs and the equilibrium density of the eigenvalues of the associated unitary ensembles were obtained in Section 4. An analysis parallel to that of Section 2-4 is carried out in Section 5 on the Hankel determinant for $x^{\alpha}\exp(-x-\sum_{k=1}^m t_k/x^k)$ with $m=3$. For general $m$, we outline in Section 6 the derivation process of the PDE satisfied by the logarithmic derivative of the Hankel determinant. Conclusions are drawn in Section 7.

\section{Difference Equations}\label{Sdiffeq}
In this section, we make use of $(S_1)$ and $(S_2')$, given in the Introduction, to express the recurrence coefficients in terms of four auxiliary quantities which are shown to satisfy a system of difference equations that can be iterated.

We now compute $A_n$ and $B_n$ from their definitions \eqref{defAn}-\eqref{defBn}, which first requires working out $\left(zv'(z)-yv'(y)\right)/(z-y)$. For our weight function given by \eqref{w}, we have
\begin{align*}
v(x)&=-\ln w(x;t_1,t_2)=-\al\:\ln\:x +x +\frac{t_1}{x}+\frac{t_2}{x^2},\\
v'(x)&=-\frac{\alpha}{x} +1 -\frac{t_1}{x^2} -\frac{2t_2}{x^3},
\end{align*}
so that
\begin{align*}
\frac{zv'(z)-yv'(y)}{z-y}=&1+\frac{1}{z}\left(\frac{t_1}{y}+\frac{2t_2}{y^2}\right)+\frac{1}{z^2}\cdot\frac{2t_2}{y}.
\end{align*}
Substituting it into \eqref{defAn}-\eqref{defBn}, we get the following expressions for  $A_n$ and $B_n$.

\begin{lemma}
We have
\begin{align}
A_n(z)=&\frac{1}{z} + \frac{R_n+R_n^{\star}}{z^2} +\frac{\tau R_n}{z^3},\label{An-2}\\
B_n(z)=&-\frac{n}{z}+ \frac{r_n+r_n^{\star}}{z^2} + \frac{\tau r_n}{z^3},\label{Bn-2}
\end{align}
where $\tau:=2t_2/t_1$ and
\begin{subequations}\label{defRr}
\begin{align}
R_n(t_1,t_2):=&\frac{t_1}{h_n}\:\int_{0}^{+\infty}P_n^2(y;t_1,t_2)w(y;t_1,t_2)\frac{dy}{y},\label{defR}\\
R_n^{\star}(t_1,t_2):=&\frac{2t_2}{h_n}\:\int_{0}^{+\infty}P_n^2(y;t_1,t_2)w(y;t_1,t_2)\frac{dy }{y^2},\label{defR*}\\
r_n(t_1,t_2):=&\frac{t_1}{h_{n-1}}\:\int_{0}^{+\infty}P_n(y;t_1,t_2)P_{n-1}(y;t_1,t_2)w(y;t_1,t_2)\frac{dy }{y},\label{defr}\\
r_n^{\star}(t_1,t_2):=&\frac{2t_2}{h_{n-1}}\:\int_{0}^{+\infty} P_n(y;t_1,t_2)P_{n-1}(y;t_1,t_2)w(y;t_1,t_2)\frac{dy}{y^2}.\label{defr*}
\end{align}
\end{subequations}
\end{lemma}

Plugging \eqref{An-2}-\eqref{Bn-2} into $(S_1)$ and $(S_2')$, and equating the residues on their both sides, we obtain a set of equations. Comparing the coefficients of $z^{-j}$ (for $j=1,2,3$) on both sides of $(S_1)$ gives us
\begin{gather}
\alpha_{n}=2n+1+\alpha+R_n+R_n^{\star},\label{S1eq1}\\
r_{n+1}+r_n+r_{n+1}^{\star}+r_n^{\star}=\tau R_n-\alpha_n(R_n+R_n^{\star})+t_1,\label{S1eq2}\\
\tau(r_{n+1}+r_n)=2t_2-\tau\alpha_nR_n.\nonumber
\end{gather}
Since $t_2\neq0$,  we divide both sides of the last equation by $\tau$ and get
\begin{gather}\label{S1eq4}
r_{n+1}+r_n=t_1-\alpha_nR_n.
\end{gather}
Inserting it into \eqref{S1eq2} yields
\begin{gather}\label{S1eq5}
r_{n+1}^{\star}+r_n^{\star}=\tau R_n-\alpha_nR_n^{\star}.
\end{gather}

By equating the coefficients of $z^{-j}$ ($j=2,3,\dots,6$) in $(S_2')$, we find
\begin{gather}
\beta_{n}=n(n+\alpha)+r_n+r_n^{\star}+\sum_{j=0}^{n-1}R_{j}+\sum_{j=0}^{n-1}R_{j}^{\star},\label{S2'eq1}\\
\beta_{n}\left(R_{n-1}+R_n+R_{n-1}^{\star}+R_n^{\star}\right)=-(2n+\alpha)( r_n+r_n^{\star})+\tau r_n+nt_{1}+\tau\sum_{j=0}^{n-1}R_{j},\nonumber\\
\begin{aligned}\label{S2'eq3}
\beta_{n}&\left(\tau(R_n+R_{n-1})+(R_n+R_n^{\star})( R_{n-1}+R_{n-1}^{\star})\right)\\
&\qquad=(r_n+r_n^{\star})^{2}-t_{1}\left(r_n+r_n^{\star}\right)-(2n+\alpha)\tau\, r_n+2nt_2,
\end{aligned}
\end{gather}
\vspace{-7mm}
\begin{gather}\label{S2'eq4}
\tau\beta_{n}\left(R_{n-1}(R_n+R_n^{\star})+R_n(R_{n-1}+R_{n-1}^{\star})\right)=2\tau\, r_n\left(r_n+r_n^{\star}\right)- 2t_2(2 r_n+r_n^{\star}),\\
\tau^2\beta_{n}R_nR_{n-1}=\tau^2r_n(r_n-t_1).\label{S2'eq5}
\end{gather}
Dividing \eqref{S2'eq4} by $\tau$ and \eqref{S2'eq5} by $\tau^2$ leads us to
\begin{gather}
\beta_{n}\left(2R_nR_{n-1}+R_{n-1}R_n^{\star}+R_nR_{n-1}^{\star}\right)
=2r_n\left(r_n+r_n^{\star}\right)-t_1(2r_n+r_n^{\star}),\nonumber\\
\beta_{n}R_nR_{n-1}=r_n(r_n-t_1).\label{S2'eq6}
\end{gather}
Combining these two equations to get rid of $\beta_nR_nR_{n-1}$ produces
\bea\label{S2'eq7}
\beta_n(R_{n-1}R_n^{\star}+R_nR_{n-1}^{\star})=r_n^{\star}(2r_n-t_1).
\eea
Substituting \eqref{S2'eq6}-\eqref{S2'eq7} into \eqref{S2'eq3} results in
\bea\label{S2'eq8}
\beta_n(\tau(R_n+R_{n-1})+R_n^{\star}R_{n-1}^{\star})=\left(r_n^{\star}\right)^2-(2n+\alpha)\tau\,r_n+2nt_2.
\eea
Combining  \eqref{S2'eq6} and \eqref{S2'eq7} to eliminate $R_{n-1}$, we come to
\bea\label{betanRn-1*}
\beta_{n}R_{n-1}^{\star}=\frac{r_n^{\star}(2r_n-t_1)}{R_n}+r_n(t_1-r_n)\frac{R_n^{\star}}{R_n^2}.
\eea

Using \eqref{betanRn-1*} and \eqref{S2'eq6} to remove $R_{n-1}^{\star}$ and $R_{n-1}$ respectively in
\eqref{S2'eq8}, we obtain an expression for $\beta_n$ involving the auxiliary quantities. This expression and \eqref{S1eq1} are stated in the following lemma.

\begin{lemma} The recurrence coefficients are expressed in terms of the auxiliary quantities by
\begin{gather}\label{alphanRn}
\alpha_{n}=2n+1+\alpha+R_n+R_n^{\star},
\end{gather}
\begin{equation}\label{betanRnrn}
\begin{aligned}
\beta_n=&\frac{1}{\tau R_n}\left(r_n^{\star}-r_n\frac{R_n^{\star}}{R_n}\right)\left(r_n^{\star}+(t_1-r_n)\frac{R_n^{\star}}{R_n}\right)+\frac{r_n(t_1-r_n)}{R_n^2}+\frac{nt_1-(2n+\alpha)r_n}{R_n}.
\end{aligned}
\end{equation}
\end{lemma}

Moreover, we show that the auxiliary quantities satisfy a system of first order difference equations.

\begin{proposition} The auxiliary quantities satisfy the following difference equations
\begin{gather}
r_{n+1}=-r_n+t_1-(2n+1+\alpha+R_n+R_n^{\star})R_n,\qquad \quad n\geq0,\label{Rnrndiffeq1}\\
r_{n+1}^{\star}=-r_n^{\star}+\tau\,R_n-(2n+1+\alpha+R_n+R_n^{\star})R_n^{\star},\qquad \quad n\geq0,\label{Rnrndiffeq2}
\end{gather}
\begin{equation}\label{Rnrndiffeq3}
\begin{aligned}
R_n\left[\left(\frac{\left(r_n^{\star}\right)^2}{\tau}-(2n+\alpha)r_n+nt_1\right)R_{n-1}^2+ r_n(t_1-r_n)R_{n-1}\qquad\qquad\qquad\qquad\qquad\quad\;\right.\\
\left.+\left(r_n^{\star}(t_1-2r_n)R_{n-1}+r_n(r_n-t_1)R_{n-1}^{\star}\right)\frac{R_{n-1}^{\star}}{\tau}\right]= r_n(r_n-t_1)R_{n-1}^2,\qquad n\geq1,
\end{aligned}
\end{equation}
\begin{gather}
\qquad r_n(r_n-t_1)R_{n-1}R_n^{\star}=\left(r_n^{\star}(2r_n-t_1)R_{n-1}+r_n(t_1-r_n)R_{n-1}^{\star}\right)R_n,\qquad \quad n\geq1,\label{Rnrndiffeq4}
\end{gather}
which can be iterated in $n$ with the initial conditions
\[R_{0}=t_1\frac{\int_{0}^{+\infty}y^{\alpha-1}{\emph e}^{-y-t_1/y-t_2/y^2}dy}{\int_{0}^{+\infty}y^{\alpha}{\emph e}^{-y-t_1/y-t_2/y^2}dy},
\qquad R_{0}^{\star}=2t_2\frac{\int_{0}^{+\infty}y^{\alpha-2}{\emph e}^{-y-t_1/y-t_2/y^2}dy}{\int_{0}^{+\infty}y^{\alpha}{\emph e}^{-y-t_1/y-t_2/y^2}dy},
\qquad r_{0}=r_{0}^{\star}=0.\]
\end{proposition}
\begin{proof}
Plugging \eqref{S1eq1} into \eqref{S1eq4}-\eqref{S1eq5} gives us \eqref{Rnrndiffeq1}-\eqref{Rnrndiffeq2} immediately.
Using \eqref{S2'eq6} to eliminate $\beta_n$ in \eqref{S2'eq7} and \eqref{S2'eq8}, we get \eqref{Rnrndiffeq4} and
\begin{align*}
r_n(r_n-t_1)\left(\tau(R_n+R_{n-1})+ R_n^{\star}R_{n-1}^{\star}\right)
=&\left(\left(r_n^{\star}\right)^{2}-(2n+\alpha)\tau r_n+2nt_2\right)R_nR_{n-1},
\end{align*}
respectively. Eliminating $R_n^{\star}$ in the above equation via \eqref{Rnrndiffeq4}, we are led to \eqref{Rnrndiffeq3}.
\end{proof}

We make a remark on the special case $t_2\rightarrow0^+$. The weight function \eqref{w} becomes $x^{\al}{\rm e}^{-x-t_1/x}$ where $t_1>0$ is assumed, and it was studied in \cite{ChenIts2010}.
\begin{remark}
When $t_2\rightarrow0^+$, i.e. $\tau\rightarrow0$, we have by the definitions of $R_n^{\star}$ and $r_n^{\star}$ given by \eqref{defR*} and \eqref{defr*} that
\[R_n^{\star}\rightarrow0, \qquad\qquad r_n^{\star}\rightarrow0,\]
and
\begin{align*}
\frac{R_n^{\star}}{\tau}=&\frac{t_1}{h_n}\:\int_{0}^{\infty}P_n^2(y)w(y)\frac{dy }{y^2},& \frac{r_n^{\star}}{\tau}=&\frac{t_1}{h_{n-1}}\:\int_{0}^{\infty} P_n(y)P_{n-1}(y)w(y)\frac{dy}{y^2}.
\end{align*}
Consequently, we find that equations \eqref{alphanRn} and \eqref{betanRnrn} are consistent with (2.9) and (2.14) of \cite{ChenIts2010}
where the symbols $s,a_n$ and $b_n$ were used instead of $t_1,R_n$ and $r_n$ respectively. In addition, the difference equations \eqref{Rnrndiffeq1} and \eqref{Rnrndiffeq3}
accord with (2.16)-(2.17) of \cite{ChenIts2010}.
\end{remark}

We close this section by presenting two expressions for the coefficient of $x^{n-1}$ in $P_n(x;t_1,t_2)$, i.e. $p(n,t_1,t_2)$, which we will see play a crucial role in our subsequent derivations.
\begin{lemma} We have
\begin{align}
p(n,t_1,t_2)=&-n(n+\alpha)-\sum_{j=0}^{n-1}R_{j}-\sum_{j=0}^{n-1}R_{j}^{\star}\label{betasumR}\\
=&r_n+r_n^{\star}-\beta_n\label{betap}.
\end{align}
\end{lemma}
\begin{proof}
Replacing $n$ by $j$ in \eqref{alphanRn} and summing over $j$ from $0$ to $n-1$, in view of \eqref{a01}, we come to \eqref{betasumR}. Combining it with \eqref{S2'eq1}, we readily get \eqref{betap}.
\end{proof}

\section{Generalized Toda, Generalized Painlev\'{e} III$'$, and Generalized $\sigma$-Form of Painlev\'{e} III$'$}
We have established the difference equations for $\{R_{n},R_{n}^{\star},r_{n},r_{n}^{\star}\}$ in the preceding section. Now we turn to the derivation of differential equations for them.

We start by differentiating the orthogonality relations satisfied by $P_n(x;t_1,t_2)$, i.e. \eqref{or}, with respect to $t_1$ and $t_2$. We find that $\{R_n,R_n^{\star}\}$ and $\{r_n,r_n^{\star}\}$ are intimately connected with the differentiations of $h_n(t_1,t_2)$ and $p(n,t_1,t_2)$ respectively. With the aid of the identities derived in Section \ref{Sdiffeq}, we develop a system of first order PDEs satisfied by $\{R_{n},R_{n}^{\star},r_{n},r_{n}^{\star}\}$, which enables us to deduce two coupled second order PDEs for  $\{R_n, R_n^{\star}\}$.

For ease of notations, we denote $\frac{\partial}{\partial t_i},\frac{\partial^2}{\partial t_i^2}$ and $\frac{\partial}{\partial{t_i}}(\frac{\partial}{\partial{t_j}})$ by $\partial_{t_i}, \partial_{t_it_i}$ and $\partial_{t_it_j}$  respectively for $i,j=1,2$ in the rest of this paper.
\subsection{Generalized Toda Equation}
According to \eqref{or}, we have
\[
h_n(t_1,t_2)=\int_{0}^{+\infty}P_{n}^2(x;t_1,t_2)w(x;t_1,t_2)dx,\]
and
\[
0=\int_{0}^{+\infty}P_{n}(x;t_1,t_2)P_{n-1}(x;t_1,t_2)w(x;t_1,t_2)dx.\]
Differentiating both sides of these two identities with respect to $t_1$ and $t_2$, in view of the definition of $\{P_n(x;t_1,t_2)\}$ given by \eqref{or}-\eqref{Pndef}, we arrive at the following differential relations.

\begin{lemma}
\begin{itemize}
\item [\emph {(i)}] The derivatives of $h_n$ are related to $R_{n}$ and $R_n^{\star}$ as follows:
\begin{align}\label{hnD1}
t_1{\partial_{t_1}}\ln{h_{n}}=&- R_n,&
2t_2{\partial_{t_2}}\ln{h_{n}}=&-R_n^{\star},
\end{align}
which, in view of $\beta_n=h_n/h_{n-1}$, give us
\begin{align}\label{lnbetaD1}
t_1\partial_{t_1}\ln\beta_{n}=&R_{n-1}-R_n,&
2t_2\partial_{t_2}\ln\beta_{n}=&R_{n-1}^{\star}-R_n^{\star}.
\end{align}

\item[\emph{(ii)}] The derivatives of $p(n,t_1,t_2)$ have the following relationships with $r_{n}$ and $r_n^{\star}:$
\begin{align}\label{pD1}
t_1\partial_{t_1}p(n,t_1,t_2)=& ~r_n,&
2t_2\partial_{t_2}p(n,t_1,t_2)=& ~r_n^{\star}.
\end{align}
Since $\alpha_n=p(n,t_1,t_2)-p(n+1,t_1,t_2)$, we have
\begin{align}\label{alphaD1}
t_1\partial_{t_1}\alpha_n=&r_n-r_{n+1},&
2t_2\partial_{t_2}\alpha_n=&r_n^{\star}-r_{n+1}^{\star}.
\end{align}
\end{itemize}
\end{lemma}

Equations \eqref{lnbetaD1} and \eqref{alphaD1} provide the connections between the first order derivatives of the recurrence coefficients and the four auxiliary quantities. Combining them with \eqref{S1eq1} and \eqref{S2'eq1}, we obtain two-variable analogues of the Toda equations for $\alpha_n(t_1,t_2)$ and $\beta_n(t_1,t_2)$.

\begin{proposition}
The recurrence coefficients satisfy the following partial differential-difference equations
\begin{align}
\left(t_1\partial_{t_1}+2t_2\partial_{t_2}\right)\alpha_n
&=\beta_{n}-\beta_{n+1}+\alpha_n,\label{alphaDbeta}\\
\left(t_1\partial_{t_1}+2t_2\partial_{t_2}\right)\ln\beta_n
&=\alpha_{n-1}-\alpha_{n}+2,\label{betaDalpha}
\end{align}
which, in combination, give rise to a two-variable generalization of the Toda molecule equation for $\beta_n:$
\bea\label{betaDD}
\left(t_{1}^{2}\partial_{t_1t_1}+2t_{1}t_{2}(\partial_{t_1t_2}+\partial_{t_2t_1})+4t_{2}^{2}\partial_{t_2t_2}+2t_2\partial_{t_2}\right)\ln{\beta_{n}}
=\beta_{n-1}-2\beta_{n}+\beta_{n+1}-2.
\eea
\end{proposition}
\begin{proof}
It follows from \eqref{alphaD1} that
\begin{gather*}
\left(t_1\partial_{t_1}+2t_2\partial_{t_2}\right)\alpha_n=r_n-r_{n+1}+r_n^{\star}-r_{n+1}^{\star}.
\end{gather*}
According to \eqref{S2'eq1}, and taking \eqref{alphanRn} into account, we get
\begin{gather*}
\beta_n-\beta_{n+1}=-\alpha_n+r_n-r_{n+1}+r_n^{\star}-r_{n+1}^{\star}.
\end{gather*}
Subtracting this equation from the previous one leads us to \eqref{alphaDbeta}.

Equation \eqref{betaDalpha} can be deduced via a similar argument. From \eqref{lnbetaD1}, we have
\[\left(t_1\partial_{t_1}+2t_2\partial_{t_2}\right)\ln\beta_n
=R_{n-1}-R_n+R_{n-1}^{\star}-R_n^{\star}.\]
From \eqref{S1eq1}, we find
\begin{align*}
\alpha_{n-1}-\alpha_n=&-2+R_{n-1}-R_n+R_{n-1}^{\star}-R_n^{\star}.
\end{align*}
Combining these two equalities, we come to \eqref{betaDalpha}.

To continue, we apply $t_1\partial_{t_1}+2t_2\partial_{t_2}$ to \eqref{betaDalpha} and obtain, with the aid of \eqref{alphaDbeta}, that
\begin{align*}
\left(t_1\partial_{t_1}+2t_2\partial_{t_2}\right)^2\ln\beta_n
=\beta_{n-1}-2\beta_{n}+\beta_{n+1}+\alpha_{n-1}-\alpha_n.
\end{align*}
Writing the left-hand side out and using \eqref{betaDalpha} to get rid of $\alpha_{n-1}-\alpha_n$ on the right-hand side, we arrive at \eqref{betaDD}.
\end{proof}
\subsection{Coupled System Reducible to Painlev\'{e} III$'$}
We observe from equations \eqref{alphaD1} and \eqref{lnbetaD1} that the first order partial derivatives of $\alpha_n$ and $\beta_n$ are strongly related to the first order differences of the four auxiliary quantities. In these equations, we replace $\alpha_n$ and $\beta_n$ on the left-hand sides using \eqref{alphanRn} and \eqref{betap} respectively, and remove the terms with index $n-1$ and $n+1$ on the right-hand sides using \eqref{S1eq4}-\eqref{S1eq5}, \eqref{S2'eq6} and \eqref{betanRn-1*}. We finally produce a system of Riccati-like equations.

\begin{lemma}
The auxiliary quantities $\{R_n,R_n^{\star},r_n,r_n^{\star}\}$ satisfy the following PDE system
\begin{gather}
t_1\partial_{t_1}(R_n+R_n^{\star})=2r_n+\left(2n+1+\alpha+R_n+R_n^{\star}\right)R_n-t_1,\label{rnRD}\\
2t_2\partial_{t_2}(R_n+R_n^{\star})=2r_n^{\star}+\left(2n+1+\alpha+ R_n+R_n^{\star}\right)R_n^{\star}-\tau R_n,\label{rn*RD}\\
t_1\partial_{t_1}(r_n+r_n^{\star})=\Theta_n+r_n+\frac{2r_n(r_n-t_1)}{R_n},\label{Ric1}\\
2t_2\partial_{t_2}(r_n+r_n^{\star})
=\frac{R_n^{\star}}{R_n}\Theta_n+r_n^{\star}+ \frac{r_n^{\star}(2r_n-t_1)}{R_n},\label{Ric2}
\end{gather}
where $\Theta_n:=\frac{1}{\tau}\left(\frac{R_n^{\star}}{R_n}r_n-r_n^{\star}\right)\left(\frac{R_n^{\star}}{R_n}(t_1-r_n)+r_n^{\star}\right)+(2n+\alpha)r_n-nt_1$.
\end{lemma}
\begin{proof}
Inserting \eqref{alphanRn} into \eqref{alphaD1} yields
\begin{align*}
t_1\partial_{t_1}(R_n+R_n^{\star})=&r_n-r_{n+1},&
2t_2\partial_{t_2}(R_n+R_n^{\star})=&r_n^{\star}-r_{n+1}^{\star}.
\end{align*}
Using \eqref{S1eq4} and \eqref{S1eq5} to eliminate $r_{n+1}$ and $r_{n+1}^{\star}$ respectively in these two identities, in view of \eqref{alphanRn}, we come to \eqref{rnRD} and \eqref{rn*RD}.

To derive \eqref{Ric1} and \eqref{Ric2}, we rewrite \eqref{lnbetaD1} as
 \begin{align*}
t_1\partial_{t_1}\beta_{n}=&\beta_n(R_{n-1}-R_n),&
2t_2\partial_{t_2}\beta_{n}=&\beta_n(R_{n-1}^{\star}-R_n^{\star}).
\end{align*}
Replacing $\beta_n$ in them using \eqref{betap},  with the aid of \eqref{pD1}, we get
\begin{gather*}
t_1\partial_{t_1}(r_n+r_n^{\star})-r_n=\beta_nR_{n-1}-\beta_nR_n,\\
2t_2\partial_{t_2}(r_n+r_n^{\star})-r_n^{\star}=\beta_nR_{n-1}^{\star}-\beta_nR_n^{\star}.
\end{gather*}
Using \eqref{S2'eq6} and \eqref{betanRn-1*} to eliminate $\beta_nR_{n-1}$ and $\beta_nR_{n-1}^{\star}$ respectively in the above two equations, in light of \eqref{betanRnrn}, we arrive at \eqref{Ric1} and \eqref{Ric2}.
\end{proof}

Solving $r_n$ and $r_n^{\star}$ from \eqref{rnRD} and \eqref{rn*RD}, and substituting them into \eqref{Ric1}-\eqref{Ric2}, we derive a pair of second order PDEs for $R_n$ and $R_n^{\star}$.
\begin{theorem}\label{RR*PDE}
The quantities $R_n(t_1,t_2)$ and $R_n^{\star}(t_1,t_2)$ satisfy the following coupled second order PDEs:
\begin{subequations}\label{RnRn*PDEs}
\begin{equation}\label{RnRn*PDE1}
\begin{aligned}
&t_1^2\left(\partial_{t_1t_1}S_n\right)+2t_1t_2\left(\partial_{t_1t_2}S_n\right)+t_1t_2\left(\frac{T_n}{2t_2} \cdot t_1\partial_{t_1}S_n-\partial_{t_2}S_n\right)^2\\
&-\frac{t_1^2}{R_n}(\partial_{t_1}S_n)^2+(1-R_n^{\star})\cdot t_1\partial_{t_1}S_n+2t_2\left(R_n
   \cdot\partial_{t_2}S_n+\partial_{t_1}R_n\right)\\
&-R_nS_n(S_n+2n+1+\alpha)+\frac{t_2 }{t_1}R_n(R_n-2)-\alpha t_1 +\frac{t_1^2}{R_n}-\frac{t_1^3}{4t_2}T_n^2=0,
\end{aligned}
\end{equation}
\begin{equation}\label{RnRn*PDE2}
\begin{aligned}
&4t_2^2
  \left(\partial_{t_2t_2}S_n\right)+2 t_1t_2
   \left(\partial_{t_2t_1}S_n\right)+\frac{t_1}{4t_2 }T_n\left(T_n\cdot t_1\partial_{t_1}S_n-2t_2\partial_{t_2}S_n\right)^2\\
&-\frac{2t_1 t_2}{R_n}
  \cdot \partial_{t_1}S_n\cdot\partial_{t_2}S_n+\partial_{t_1}S_n\left(t_1 \left(T_n+R_n^{\star}\right)-2 t_2\right)+(1-R_n)\cdot 2t_2\partial_{t_2}S_n+\frac{4t_2^2}{t_1}\partial_{t_2}R_n\\
&+\left(\frac{2t_2}{t_1}R_n-R_n^{\star}S_n\right)(S_n +2n+1+\alpha)+\frac{t_2}{t_1}R_n\left(R_n^{\star}+2\right)-\alpha t_1 T_n-\frac{t_1^3}{4t_2}T_n^3=0,
\end{aligned}
\end{equation}
\end{subequations}
where
\begin{align*}
S_n(t_1,t_2):=&R_n(t_1,t_2)+R_n^{\star}(t_1,t_2),&T_n(t_1,t_2):=&\frac{R_n^{\star}(t_1,t_2)}{R_n(t_1,t_2)}.
\end{align*}
\end{theorem}

\begin{remark}\label{0odeRn}
When $t_2\rightarrow0^{+}$ and $t_1>0$, by the definitions of $R_n$ and $R_n^{\star}$ given by \eqref{defR}-\eqref{defR*}, we have
\begin{align*}
R_n^{\star}\rightarrow0,&\qquad\quad
S_n\rightarrow R_n,\qquad\quad T_n\rightarrow0,& \frac{T_n}{2t_2}=\frac{\int_{0}^{+\infty}\frac{dy }{y^2}P_n^2(y;t_1,t_2)w(y;t_1,t_2)}{t_1\int_{0}^{+\infty}\frac{dy}{y}P_n^2(y;t_1,t_2)w(y;t_1,t_2)}.
\end{align*}
Hence, it follows from \eqref{RnRn*PDE1} that
\begin{equation}\label{Rode}
R_n''=\frac{(R_n')^2}{R_n}-\frac{R_n'}{t_1}+\frac{R_n^3}{t_1^2}+(2n+1+\alpha)\frac{R_n^2}{t_1^2}+\frac{\alpha}{t_1} -\frac{1}{R_n},
\end{equation}
where the derivative is with respect to $t_1$, and \eqref{RnRn*PDE2} is reduced to $0=0$. Equation \eqref{Rode} coincides with Theorem 1 of \cite{ChenIts2010} where $s$ and $a_n$ correspond to our $t_1$ and $R_n$ respectively.
As is pointed out therein, \eqref{Rode} can be transformed into $P_{III'}\left(-4(2n+1+\alpha),-4\alpha,4,-4\right)$ satisfied by $-R_n$.
\end{remark}

\subsection{Generalized $\sigma$-Form of Painlev\'{e} $III^\prime$}
We turn our attention to the Hankel determinant $D_n(t_1,t_2)$ given by \eqref{HD}. Define
\begin{gather}\label{defHn}
H_{n}\left(t_{1},t_{2}\right):=\left(t_{1}\partial_{t_1}+2t_{2}\partial_{t_2}\right)\ln{D_{n}(t_{1},t_{2})}.
\end{gather}
Recall \eqref{Dnhn}, i.e.
$
D_{n}(t_{1},t_{2})=\prod_{j=0}^{n-1}h_{j}(t_{1},t_{2}).
$
Applying $t_1\partial_{t_1}+2t_2\partial_{t_2}$ to it, in view of \eqref{hnD1} and \eqref{betasumR}, we get the relation between $H_n$ and $p(n,t_1,t_2)$, i.e. the coefficient of $x^{n-1}$ in $P_n(x;t_1,t_2)$.
\begin{lemma} We have
\begin{gather}\label{Hnpn}
H_{n}\left(t_{1},t_{2}\right)=n(n+\alpha)+p(n,t_1,t_2),
\end{gather}
which, on account of \eqref{pD1}, results in
\begin{align}\label{HnD}
t_1\partial_{t_1}H_n(t_1,t_2)=&r_n(t_1,t_2),&
2t_2\partial_{t_2}H_n(t_1,t_2)=&r_n^{\star}(t_1,t_2).
\end{align}
\end{lemma}

Replacing $p(n,t_1,t_2)$ in \eqref{betap} using \eqref{Hnpn} yields
\begin{align}
\beta_n=&r_n+r_n^{\star}-H_n+ n(n+\alpha)\label{betanrnHn}\\
=&\left(t_1\partial_{t_1}+2t_2\partial_{t_2}\right)H_n-H_n+n(n+\alpha),\label{betanHnD}
\end{align}
where the second equality is due to \eqref{HnD}.

\begin{remark}
Inserting \eqref{defHn} into \eqref{betanHnD} leads to \begin{align}\label{btR3}
\beta_n=\mathcal{L}\ln D_{n}+n(n+\alpha),
\end{align}
where $\mathcal{L}:=t_1^2\partial_{t_1t_1}+4t_2^2\partial_{t_2t_2}+2t_1t_2(\partial_{t_1t_2}+\partial_{t_2t_1})+2t_2\partial_{t_2}$.
According to \eqref{a12} and \eqref{Dnhn}, we find
\begin{align}\label{btDn}
\beta_n=\frac{D_{n+1}D_{n-1}}{D_n^2}.
\end{align}
Combining it with \eqref{btR3} gives us the following second order partial differential-difference equation for $D_n(t_1,t_2):$
\[
\mathcal{L}\ln D_n=\frac{D_{n+1}D_{n-1}}{D_{n}^{2}}-n(n+\alpha).
\]
In addition, from \eqref{btDn}, we get
\[\mathcal{L}\ln \beta_n=\mathcal{L}\ln D_{n+1}+\mathcal{L}\ln D_{n-1}-2\mathcal{L}\ln D_n.\]
Using \eqref{btR3} to replacing the three terms on the right-hand side, we arrive at \eqref{betaDD}.
\end{remark}

Combining \eqref{betanrnHn} with \eqref{betanRnrn}, and replacing $r_n$ and $r_n^{\star}$ in the resulting expression using \eqref{rnRD}-\eqref{rn*RD}, we express $H_n$ in terms of $\{R_n,R_n^{\star}\}$ and their first order derivatives. Conversely, by using \eqref{lnbetaD1}, \eqref{S2'eq6} and \eqref{betanRn-1*}, $R_n$ and $R_n^{\star}$ are expressed in terms of $H_n$ and its derivatives. Substituting these expressions into the one for $H_n$, we finally derive the PDE for $H_n$.
\begin{theorem}\label{ThPDEHn}
$H_n$ is expressed in terms of $R_n,R_n^{\star}$ and their first order partial derivatives by
\begin{equation}\label{HnRnRn*D}
\begin{aligned}
H_n(t_1,t_2)
=&-\frac{t_1}{8t_2R_n}\left(\frac{R_n^{\star}}{R_n}t_1\partial_{t_1}S_n-2t_2\partial_{t_2}S_n\right)^2
+\frac{1}{4}\left(\frac{t_1\partial_{t_1}S_n }{R_n}-1\right)^2\\
&-\frac{S_n^2}{4}-\left(n+\frac{\alpha}{2}\right)S_n+\frac{t_2}{2t_1} R_n+\frac{t_1}{2}-\frac{1}{4}\left(\frac{t_1}{R_n}-\alpha\right)^2+\frac{t_1^3}{8t_2}\cdot\frac{ (R_n^{\star})^2}{R_n^3},
\end{aligned}
\end{equation}
where $S_n=R_n+R_n^{\star}$. Moreover,
$H_n$ satisfies a second order sixth degree PDE:
\begin{equation}\label{PDEHn}
\begin{aligned}
&\left(\left(\partial_{t_1}\beta_n\right)^2+4\beta_n\left(\partial_{t_1}H_n\right)\left(\partial_{t_1}H_n-1\right)\right)^3\\
=&\left[\left(\partial_{t_1}\beta_n\right)^2\cdot\left(-2t_2(\partial_{t_2}H_n)^2+(2 n+\alpha )\left(\partial_{t_1}H_n\right) -n\right)\right.\\
&\left.\quad+2t_2\cdot\partial_{t_2}\beta_n\cdot\big(\partial_{t_1}\beta_n\cdot\left(\partial_{t_2}H_n\right)\left(2\partial_{t_1}H_n-1\right)-\partial_{t_2}\beta_n\cdot\left(\partial_{t_1}H_n\right)\left(\partial_{t_1}H_n-1\right)\big)\right.\\
&\left.\quad+2\beta_n\big(2\left(\partial_{t_1}H_n\right)\left(\partial_{t_1}H_n-1\right)\left((2n+\alpha)\partial_{t_1}H_n-n\right)+t_2(\partial_{t_2}H_n)^2\big)\right]^2,
\end{aligned}
\end{equation}
where $\beta_n$ is given by \eqref{betanHnD} from which follows the expressions for $\partial_{t_1}\beta_n$ and $\partial_{t_2}\beta_n:$
\begin{align}\label{DtbtH}
\partial_{t_1}\beta_n=&\left(t_1\partial_{t_1t_1}+2t_2\partial_{t_1t_2}\right)H_n,&
\partial_{t_2}\beta_n=&\left(t_1\partial_{t_2t_1}+2t_2\partial_{t_2t_2}+\partial_{t_2}\right)H_n.
\end{align}
\end{theorem}

\begin{proof}
We rewrite \eqref{lnbetaD1} as
\begin{align*}
t_1\partial_{t_1}\beta_{n}=&\beta_{n}R_{n-1}-\beta_{n}R_n,&
2t_2\partial_{t_2}\beta_{n}=&\beta_{n}R_{n-1}^{\star}-\beta_{n}R_n^{\star}.
\end{align*}
Using \eqref{S2'eq6} and \eqref{betanRn-1*} to clear $\beta_nR_{n-1}$ and $\beta_nR_{n-1}^{\star}$ respectively in these two equations yields
\bea\label{eq2Rn}
t_1\partial_{t_1}\beta_{n}=\frac{r_n(r_n-t_1)}{R_n}-\beta_{n}R_n,
\eea
which is quadratic in $R_n$, and
\bea\label{eq2Rn*}
2t_2\partial_{t_2}\beta_{n}=\frac{r_n^{\star}(2r_n-t_1)}{R_n}-r_n(r_n-t_1)\frac{R_n^{\star}}{R_n^2}-\beta_{n}R_n^{\star}.
\eea
Solving for $R_n$ from \eqref{eq2Rn} gives us the following two possible solutions
\begin{align}\label{Rnsol0}
R_n=\frac{1}{2\beta_n}\left(-t_1\partial_{t_1}\beta_n\pm\sqrt{\Delta_n}\right),
\end{align}
where
\begin{align}\label{defDelta}
\Delta_n(t_1,t_2):=\left(t_1\partial_{t_1}\beta_n\right)^2+4\beta_n r_n(r_n-t_1).
\end{align}
Here note that, by plugging \eqref{eq2Rn} into $\Delta_n$, we find
\begin{align}\label{Deltasign}
\Delta_n(t_1,t_2)=\left(\frac{r_n(r_n-t_1)}{R_n}+\beta_{n}R_n\right)^2\geq0.
\end{align}
Now we determine the sign in front of $\sqrt{\Delta_n}$ in \eqref{Rnsol0}. From \eqref{lnbetaD1}, it follows that \[R_n+\frac{t_1\partial_{t_1}\beta_n }{2\beta_n}=\frac{1}{2}\left(R_n+R_{n-1}\right).\]
According to the definition of $R_n$ given by \eqref{defR}, we know that $R_n+R_{n-1}$, and hence $R_n+\frac{t_1\partial_{t_1}\beta_n }{2\beta_n}$, has the same sign as $t_1$. This fact combined with \eqref{Rnsol0} indicates that
\begin{gather}\label{Rnsol}
R_n=\frac{1}{2\beta_n}\left(-t_1\partial_{t_1}\beta_n+{\rm sgn} (t_1)\cdot\sqrt{\Delta_n}\right),
\end{gather}
where ${\rm sgn}(t_1)$ is the sign function of $t_1$, which is 1 for $t_1>0$ and $-1$ for $t_1<0$.
Next, we solve $R_n^{\star}$ from \eqref{eq2Rn*} to produce
\begin{align}\label{Rn*}
R_n^{\star}=\frac{r_n^{\star}(2r_n-t_1)-2t_2\partial_{t_2}\beta_n\cdot R_n}{\frac{r_n(r_n-t_1)}{R_n}+\beta_nR_n}.
\end{align}
Getting rid of $\frac{r_n(r_n-t_1)}{R_n}$ in the denominator using \eqref{eq2Rn}, and substituting \eqref{Rnsol} into the obtained expression, we find
\bea\label{Rn*sol} R_n^{\star}=\frac{r_n^{\star}(2r_n-t_1)+t_1t_2\frac{\partial_{t_1}\beta_n\cdot\partial_{t_2}\beta_n}{\beta_n}}{{\rm sgn}(t_1)\cdot\sqrt{\Delta_n}}-t_2\frac{\partial_{t_2}\beta_n}{\beta_n}.
\eea

To continue, we rewrite \eqref{betanRnrn} as
\begin{equation*}
1-\frac{t_1}{2t_2 \beta_nR_n}\left(r_n^{\star}-r_n\frac{R_n^{\star}}{R_n}\right)\left(r_n^{\star}+(t_1-r_n)\frac{R_n^{\star}}{R_n}\right)-\frac{r_n(t_1-r_n)}{\beta_nR_n^2}-\frac{nt_1-(2n+\alpha)r_n}{\beta_nR_n}=0.
\end{equation*}
Substituting \eqref{Rnsol} and \eqref{Rn*sol} into the left-hand side of this equation, multiplying the obtained expression by $t_2\Delta_n\left(-t_1\partial_{t_1}\beta_n+{\rm sgn} (t_1)\cdot\sqrt{\Delta_n}\right)^3$ and replacing $\{r_n,r_n^{\star}\}$ using \eqref{HnD}, in view of the fact that $\left({\rm sgn}(t_1)\right)^2=1$,  we obtain
\begin{align}\label{PDEHn-1}
\mathcal{A}(t_1,t_2)+{\rm sgn}(t_1)\cdot\mathcal{B}(t_1,t_2)\sqrt{\Delta_n(t_1,t_2)}=0,
\end{align}
where
\begin{align*}
\frac{\mathcal{A}(t_1,t_2)}{-4t_1^5t_2}:=& \left(\partial_{t_1}\beta_n\right)^2 \left[(\partial_{t_1}\beta_n)^2 \left(\partial_{t_1}\beta_n+2 t_2
   (\partial_{t_2}H_n)^2-(2n+\alpha) \partial_{t_1}H_n+n\right)+8\beta_n\Gamma_n\left(\partial_{t_1}\beta_n\right)\right]\\
   &+2 t_2\left(\left(\partial_{t_1}\beta_n\right)^2+2\beta_n\Gamma_n\right)\left[ \left(\partial_{t_2}H_n\right) \left(1-2
   \partial_{t_1}H_n\right)\left(\partial_{t_1}\beta_n\right)\left(\partial_{t_2}\beta_n\right)+\Gamma_n \left(\partial_{t_2}\beta_n\right)^2 \right]\\
   &+2 \beta_n (\partial_{t_1}\beta_n)^2 \left[t_2 (\partial_{t_2}H_n)^2 (2 \Gamma_n -1)+3\Gamma_n  \left(-(2n+\alpha ) \partial_{t_1}H_n+ n\right)\right]\\
   &+4 \beta_n^2 \Gamma_n  \left[2 \Gamma_n  \left(2
   \partial_{t_1}\beta_n-(2n+\alpha) \partial_{t_1}H_n+n\right)-t_2 (\partial_{t_2}H_n)^2\right],\\
\frac{\mathcal{B}(t_1,t_2)}{-4t_1^4t_2}:=&-\left(\partial_{t_1}\beta_n\right)
   \left[\left(\partial_{t_1}\beta_n\right)^2 \left(\partial_{t_1}\beta_n+2 t_2
   \left(\partial_{t_2}H_n\right)^2-(2n+\alpha) \partial_{t_1}H_n+n\right)\right.\\
   &\left.\qquad\qquad\qquad+2 t_2
   \left(\partial_{t_2}H_n\right) \left(1-2 \partial_{t_1}H_n\right) \left(\partial_{t_1}\beta_n\right) \left(\partial_{t_2}\beta_n\right)+2 t_2
   \Gamma_n \left(\partial_{t_2}\beta_n\right)^2\right]\\
   &+2 \beta_n  \left(\partial_{t_1}\beta_n\right) \left[\Gamma_n
   \left(-3 \partial_{t_1}\beta_n +2 (2n+\alpha) \partial_{t_1}H_n-2 n\right)+t_2 \left(\partial_{t_2}H_n\right)^2\right]-8 \beta_n^2\Gamma_n^2.
\end{align*}
Here $\Gamma_n:=\left(\partial_{t_1}H_n\right)\left(\partial_{t_1}H_n-1\right).$
Clearing the square root in \eqref{PDEHn-1}, we have
\[
\mathcal{A}^2-\mathcal{B}^2\Delta_n=0.
\]
After simplification by Mathematica and some computations by hand, we arrive at \eqref{PDEHn}.
\end{proof}
\begin{remark}
Assuming $H_n$ is independent of $t_2$, \eqref{PDEHn} is reduced to
\begin{align*}
&\left[\left(t_1H_n''\right)^2+4\left(t_1H_n'-H_n+n(n+\alpha)\right)H_n'\left(H_n'-1\right)\right]^2\\
&\qquad\cdot\big\{\left(t_1H_n''\right)^2+4\left(t_1H_n'-H_n+n(n+\alpha)\right)H_n'\left(H_n'-1\right)-\left((2 n+\alpha )H_n' -n\right)^2\big\}=0.
\end{align*}
From \eqref{defDelta}, we know that the term in the above square brackets is $\Delta_n/t_1^2$ which is shown to be nonnegative in \eqref{Deltasign}. Hence, the term in the curly brackets in the above equation is zero,
which agrees with (3.24) of \cite{ChenIts2010}. As was pointed out below (3.20) therein, $\sigma_n(t_1):=H_n-t_1/2-n(n+\alpha)/2$ is the Jimbo-Miwa-Okamoto $\sigma$-function corresponding to the Painlev\'{e} $III^{\prime}$ equation satisfied by $-R_n$ (see our Remark \ref{0odeRn}).
\end{remark}

\section{Double Scaling Analysis and Generalized Painlev\'{e} III Equation}
In this section, we proceed with the large $n$ analysis of the Hankel determinant under the assumption that $n\rightarrow\infty$ and $t_1\rightarrow0,t_2\rightarrow0^+$ such that $s_1$ and $s_2$ defined by
\begin{align}\label{ds}
s_1:=2nt_1,\qquad\qquad\qquad s_2:=4n^2t_2
\end{align}
are fixed. This double scaling is motivated as follows.
 \begin{itemize}
 \item [(i)] When $t_2\rightarrow0^+$, the weight function $w(x;t_1,t_2)$ given by \eqref{w} reads $x^{\alpha}{\rm e}^{-x-t_1/x}$. Under the scaling that $n\rightarrow\infty$ and $t_1\rightarrow0^+$ such that $s_1=2nt_1$ is fixed, the correlation kernel of the associated unitary ensemble is given by a $\Psi$-kernel \cite[Corollary 1]{XuDaiZhao14}.
     \item [(ii)]  The partition function of the unitary ensemble for the weight $y^{\alpha}{\rm e}^{-n(y+t^2/y^2)}$ was studied in  \cite{ACM} (see (1.51) therein with $k=2$). It was shown in Theorem 1.11 that, as $n\rightarrow\infty$ and $t\rightarrow0^+$ such that $cn^{5/2}t\rightarrow s$ with $c$ denoting a constant, the eigenvalue correlation kernel of this ensemble is characterized by a hierarchy of higher order Painlev\'{e} III equation.

         By a change of variable $y=x/n$, the above-mentioned partition function of \cite{ACM} is transformed into the one for $w(x;0,n^3t^2)$, up to a multiplicative constant. The aforementioned double scaling in \cite{ACM} indicates $c^2n^2t_2\rightarrow s^2$ which explains why we define $s_2:=4n^2t_2$.
     \end{itemize}

\subsection{Generalized Painlev\'{e} III Equation }\label{dsRH}
We assume that the limit of the double scaled Hankel determinant exists, defined by
\[\Delta(s_1,s_2):=\lim\limits_{n\rightarrow\infty}D_n\left(\frac{s_1}{2n},\frac{s_2}{4n^2}\right).\]
It follows from the definition of $H_n$ given by \eqref{defHn} that
\begin{align*}
\lim\limits_{n\rightarrow\infty}H_n(t_1,t_2)&=:H(s_1,s_2)=\left(s_1\partial_{s_1}+2s_2\partial_{s_2}\right)\ln\Delta(s_1,s_2).
\end{align*}
According to \eqref{HnD} and \eqref{Rnsol}-\eqref{Rn*}, we deduce the following connections between the double scaled $\{R_n,R_n^{\star},r_n,r_n^{\star}\}$ and $H(s_1,s_2)$.
\begin{proposition} \label{RRstarH}
We have
\begin{align}
\lim\limits_{n\rightarrow\infty}nR_n(t_1,t_2)=:R(s_1,s_2)=&-s_1\partial_{s_1}H(s_1,s_2)\label{dsrHR}\\
=&-\lim\limits_{n\rightarrow\infty}r_n(t_1,t_2)=:-r(s_1,s_2),\label{dsrHR-1}\\
\lim\limits_{n\rightarrow\infty}nR_n^{\star}(t_1,t_2)=:R^{\star}(s_1,s_2)=&-2s_2\partial_{s_2}H(s_1,s_2)\label{dsr*HR*}\\
=&-\lim\limits_{n\rightarrow\infty}r_n^{\star}(t_1,t_2)=:-r^{\star}(s_1,s_2).\label{dsr*HR*-1}
\end{align}
\end{proposition}
\begin{proof}
Since $s_1=2nt_1$ and $s_2=4n^2t_2$, we have
\begin{align}\label{DHnts}
\partial_{t_1}H_n=2n\cdot\partial_{s_1}H_n,\qquad\qquad \partial_{t_2}H_n=4n^2\cdot\partial_{s_2}H_n.
\end{align}
Combining them with \eqref{HnD} gives us \eqref{dsrHR-1} and \eqref{dsr*HR*-1}.

Inserting \eqref{HnD} and \eqref{DtbtH} into \eqref{Rnsol} and \eqref{Rn*sol} yields
\begin{gather}\label{Rnsol2}
R_n=\frac{1}{2\beta_n}\left(-t_1\left(t_1\partial_{t_1t_1}+2t_2\partial_{t_1t_2}\right)H_n+{\rm sgn}(t_1)\cdot\sqrt{\Delta_n}\right),
\end{gather}
and
\begin{equation}\label{Rn*sol2}
\begin{aligned} R_n^{\star}=&\frac{t_1t_2\left(2\partial_{t_2}H_n\cdot(2\partial_{t_1}H_n-1)+\frac{\left(t_1\partial_{t_1t_1}+2t_2\partial_{t_1t_2}\right)H_n\cdot\left(t_1\partial_{t_2t_1}+2t_2\partial_{t_2t_2}+\partial_{t_2}\right)H_n}{\beta_n}\right)}{{\rm sgn}(t_1)\cdot\sqrt{\Delta_n}}\\
&-\frac{t_2}{\beta_n}\left(t_1\partial_{t_2t_1}+2t_2\partial_{t_2t_2}+\partial_{t_2}\right)H_n,
\end{aligned}
\end{equation}
where $\beta_n=\left(t_1\partial_{t_1}+2t_2\partial_{t_2}\right)H_n-H_n+n(n+\alpha)$, given by \eqref{betanHnD}, and
\begin{align*}
\Delta_n=&t_1^2\left(\big(\left(t_1\partial_{t_1t_1}+2t_2\partial_{t_1t_2}\right)H_n\big)^2+4\beta_n \cdot\partial_{t_1}H_n\cdot(\partial_{t_1}H_n-1)\right),
\end{align*}
so that
\begin{align*}
{\rm sgn}(t_1)\cdot\sqrt{\Delta_n}
=&t_1\sqrt{\big(\left(t_1\partial_{t_1t_1}+2t_2\partial_{t_1t_2}\right)H_n\big)^2+4\beta_n \cdot\partial_{t_1}H_n\cdot(\partial_{t_1}H_n-1)}.
\end{align*}
From \eqref{DHnts}, it follows that
\begin{align*}
\partial_{t_1t_1}H_n=4n^2\cdot\partial_{s_1s_1}H_n,&\qquad\qquad \partial_{t_1t_2}H_n=8n^3\cdot\partial_{s_1s_2}H_n,\\
\partial_{t_2t_1}H_n=8n^3\cdot\partial_{s_2s_1}H_n,&\qquad \qquad \partial_{t_2t_2}H_n=16n^4\cdot\partial_{s_2s_2}H_n.
\end{align*}
Substituting these identities and \eqref{DHnts} into the right-hand sides of \eqref{Rnsol2} and \eqref{Rn*sol2}, replacing $t_1$ by $s_1/(2n)$ and $t_2$ by $s_2/(4n^2)$ in the resulting expressions, we obtain
\begin{align*}
R_n(t_1,t_2)=&\frac{{\rm sgn}(\partial_{s_1}H_n)\cdot s_1\partial_{s_1}H_n}{n}+O(n^{-2}),\\
R_n^{\star}(t_1,t_2)=&\frac{{\rm sgn}(\partial_{s_1}H_n)\cdot2s_2\partial_{s_2}H_n}{n}+O(n^{-2}).
\end{align*}
Hence, we have
\begin{align}
\lim\limits_{n\rightarrow\infty}nR_n(t_1,t_2)={\rm sgn}(\partial_{s_1}H)\cdot s_1\partial_{s_1}H,\label{Rnlim}\\
\lim\limits_{n\rightarrow\infty}nR_n^{\star}(t_1,t_2)={\rm sgn}(\partial_{s_1}H)\cdot2s_2\partial_{s_2}H.\label{Rn*lim}
\end{align}

Now we make use of the Riccati equation \eqref{rnRD} to determine ${\rm sgn}(\partial_{s_1}H)$. Replacing $t_2\partial_{t_2}$ by $s_2\partial_{s_2}$ in it and taking the limit as $n\rightarrow\infty$ on both sides, with the aid of \eqref{dsrHR-1}, \eqref{dsr*HR*-1} and \eqref{Rnlim}-\eqref{Rn*lim}, we find ${\rm sgn}(\partial_{s_1}H)=-1$. Consequently, \eqref{Rnlim} and \eqref{Rn*lim} reduce to \eqref{dsrHR} and \eqref{dsr*HR*}.
\end{proof}

We are now in a position to establish equations for $R$ and $R^{\star}$ by using the pair of PDEs satisfied by $R_n$ and $R_n^{\star}$, i.e. \eqref{RnRn*PDE1}-\eqref{RnRn*PDE2}.
\begin{theorem} \label{RR*PDEs}
$R(s_1,s_2)$ and $R^{\star}(s_1,s_2)$ satisfy the following coupled second order PDEs:
\begin{subequations}\label{RR*limeqs}
\begin{equation}\label{RR*limeq1}
\begin{aligned}
s_1^2\partial_{s_1s_1}U&+2s_1s_2\partial_{s_1s_2}U+2s_1s_2\left( \frac{s_1R^{\star}}{2s_2R} \cdot s_1
   \partial_{s_1}U-
   \partial_{s_2}U\right
   )^2\\
&+
   s_1\partial_{s_1}U
   \left(1-\frac{
  s_1\partial_{s_1}U}{R}\right)-2R U-\frac{s_1^3 }{8s_2}\cdot\frac{(R^{\star})^2}{R^2}-\frac{\alpha }{2}s_1+\frac{s_1^2}{4R}=0,
\end{aligned}
\end{equation}
and
\begin{equation}\label{RR*limeq2}
\begin{aligned}
4s_2^2
   \partial_{s_2s_2}U+&2s_1s_2
  \partial_{s_2 s_1}U+\frac{s_1}{2s_2}V\left( V \cdot s_1
   \partial_{s_1}U-2 s_2
   \partial_{s_2}U\right)^2-\frac{2s_1s_2}{R}\partial_{s_1}U\cdot\partial_{s_2}U\\
&+V\cdot s_1\partial_{s_1}U+2s_2
 \partial_{s_2}U-2R^{\star}U+\frac{2s_2}{s_1}R-s_1 V\left(\frac{s_1^2}{8s_2}V^2+\frac{\alpha}{2}\right)=0,
\end{aligned}
\end{equation}
\end{subequations}
where
\[U(s_1,s_2)=R^{\star}(s_1,s_2)+R(s_1,s_2),\qquad\qquad V(s_1,s_2):=\frac{R^{\star}(s_1,s_2)}{R(s_1,s_2)}.\]

\begin{proof}
Recall that $t_1=s_1/(2n)$ and $t_2=s_2/4n^2$. According to \eqref{dsrHR} and \eqref{dsr*HR*}, we write
\begin{align}\label{defxnyn}
x_n(s_1,s_2):=nR_n\left(\frac{s_1}{2n},\frac{s_2}{4n^2}\right),\qquad\qquad y_n(s_1,s_2):=nR_n^{\star}\left(\frac{s_1}{2n},\frac{s_2}{4n^2}\right),
\end{align}
so that
\begin{align}\label{xyR}
\lim\limits_{n\rightarrow\infty}x_n=R,\qquad\qquad\lim\limits_{n\rightarrow\infty}y_n=R^{\star}.
\end{align}
We have
\begin{align}\label{RxD}
R_n=\frac{x_n}{n},\qquad\qquad\partial_{t_1}R_n=2\partial_{s_1}x_n,\qquad\qquad\partial_{t_2}R_n=4n\partial_{s_2}x_n,
\end{align}
and
\begin{align*}
\partial_{t_1t_1}R_n=4n\partial_{s_1s_1}x_n,\quad
\partial_{t_1t_2}R_n=8n^2\partial_{s_1s_2}x_n,\quad
\partial_{t_2t_1}R_n=8n^2\partial_{s_2s_1}x_n,\quad
\partial_{t_2t_2}R_n=16n^3\partial_{s_2s_2}x_n.
\end{align*}
The above identities remain valid when $R_n$ and $x_n$ are replaced by $R_n^{\star} $ and $y_n$ respectively. We insert all these relations into \eqref{RnRn*PDE1} and \eqref{RnRn*PDE2}, substitute $s_1/(2n)$ for $t_1$ and $s_2/(4n^2)$ for $t_2$ in the resulting equations, and perform a series expansion on their left-hand sides for large $n$. With the help of  Mathematica, we finally obtain two equations of the form
\begin{align}\label{fngnO}
\frac{f_n}{64s_1^2n}+O(n^{-2})\equiv0,\qquad\qquad \frac{g_n}{n}+O(n^{-2})\equiv0,
\end{align}
where $f_n$ and $g_n$ are functions of $x_n,y_n$ and their derivatives. From \eqref{xyR}, we know that $f_n=O(1)$ and $g_n=O(1)$ as $n\rightarrow\infty$. This fact together with \eqref{fngnO} indicates that \[\lim\limits_{n\rightarrow\infty}f_n=0,\qquad\qquad\lim\limits_{n\rightarrow\infty}g_n=0,\]
which lead us to \eqref{RR*limeq1} and \eqref{RR*limeq2} respectively.
\end{proof}
\end{theorem}

\begin{remark} From \eqref{dsrHR}-\eqref{dsr*HR*-1}, we know that $R=-r$ and $R^{\star}=-r^{\star}$. Plugging them into \eqref{RR*limeqs}, we readily get the coupled second order PDEs satisfied by $r$ and $r^{\star}$.
\end{remark}

\begin{remark}
When $s_2\rightarrow0^+$ (and hence $t_2\rightarrow0^+$), from \eqref{dsr*HR*}, we know that
\begin{align}\label{s20R*}
R^{\star}=0,\qquad\qquad \frac{R^{\star}}{2s_2}=-\partial_{s_2}H.
\end{align}
Consequently, \eqref{RR*limeq1}  is reduced to
\begin{align*}
R''=\frac{
   (R')^2}{R}-\frac{R'}{s_1}+
   \frac{2}{s_1^2}R^2+\frac{\alpha}{2s_1}-\frac{1}{4 R},
\end{align*}
where the derivative is with respect to $s_1$.
Define $C(s_1):=2R/s_1$, then the above equation is transformed into (2.18) of \cite{ChenChen2015}, which, as pointed out in Remark 2 therein, can be converted into a Painlev\'{e} III equation satisfied by $Y(x):=\frac{x}{2}C(\frac{x^2}{8})$.
\end{remark}

 Now we deduce from \eqref{HnRnRn*D} the expression for $H$ in terms of $\{R,R^{\star}\}$ and their first order derivatives, and establish the second order PDE for $H$ by using \eqref{PDEHn}.
\begin{theorem}$H(s_1,s_2)$ is expressed in terms of $R(s_1,s_2),R^{\star}(s_1,s_2)$ and their first order derivatives by
\begin{equation}\label{HRR*}
\begin{aligned}
H=&-\frac{s_1s_2}{R}\left(\frac{s_1 R^{\star}}{2s_2R}
  \partial_{s_1}(R+R^{\star})-\partial_{s_2}(R+R^{\star})\right
   )^2+\frac{1}{4}\left(\frac{s_1 }{R}\partial_{s_1}(R+R^{\star})
   -1\right
   )^2\\
   &-(R+R^{\star})+\frac{s_1^3(R^{\star})^2}{16s_2R^3}-\frac{1}{4}\left(\frac{s_1}{2R}-\alpha\right)^2.
\end{aligned}
\end{equation}
Moreover, $H(s_1,s_2)$ satisfies the following second order second degree PDE:
\begin{equation}\label{PDEHlargen}
\begin{aligned}
&4s_2\Big(\partial_{s_2}H \left(s_1\partial_{s_1s_1}H+2 s_2
   \partial_{s_1s_2}H\right)- \partial_{s_1}H
   \left(2 s_2 \partial_{s_2s_2}H+s_1
  \partial_{s_2s_1}H+\partial_{s_2}H\right)\Big)^2\\
&+\partial_{s_1}H\left(s_1
   \partial_{s_1s_1}H+2 s_2 \partial_{s_1s_2}H\right){}^2+4\left(\partial_{s_1}H\right)^3 \left( s_1
   \partial_{s_1}H+2 s_2 \partial_{s_2}H- H\right)\\
&-\partial_{s_1}H\left(\alpha \partial_{s_1}H+\frac{1}{2}\right)^2-s_2 \left(\partial_{s_2}H\right)^2=0.
\end{aligned}
\end{equation}
\end{theorem}
\begin{proof}
Plugging \eqref{RxD} and
\[
R_n^{\star}=\frac{y_n}{n},\qquad\qquad \partial_{t_1}R_n^{\star}=2\partial_{s_1}y_n,\qquad\qquad \partial_{t_2}R_n^{\star}=4n\partial_{s_2}y_n\]
into the right-hand side of \eqref{HnRnRn*D}, replacing $t_1$ by $s_1/(2n)$ and $t_2$ by $s_2/(4n^2)$ in the obtained expression, by expanding it as $n\rightarrow\infty$, we find
\[H_n=\Phi_n+O(n^{-1}),\]
where $\Phi_n$ is a function of $x_n,y_n$ and their first order derivatives. Taking the limit as $n\rightarrow\infty$ on both sides of the above equation leads us to \eqref{HRR*}.

Writing $\varphi_n(s_1,s_2):=H_n\left(\frac{s_1}{2n},\frac{s_2}{4n^2}\right)$,
we have
\begin{align*}
\partial_{t_1}H_n=&2n\partial_{s_1}\varphi_n,&\partial_{t_2}H_n=&4n^2\partial_{s_2}\varphi_n,\\
\partial_{t_1t_1}H_n=&4n^2\partial_{s_1s_1}\varphi_n,&\partial_{t_1t_2}H_n=&8n^3\partial_{s_1s_2}\varphi_n,\\
\;\partial_{t_2t_1}H_n=&8n^3\partial_{s_2s_1}\varphi_n,&\partial_{t_2t_2}H_n=&16n^4\partial_{s_2s_2}\varphi_n.
\end{align*}
Inserting them into the left-hand side of \eqref{PDEHn}, substituting $s_1/(2n)$ for $t_1$ and $s_2/(4n^2)$ for $t_2$ in the resulting expression, by performing an asymptotic expansion for large $n$, we obtain
\[q_nn^{10}+O(n^9)\equiv0,\]
where $q_n$ is a function of $\varphi_n$ and its derivatives. Hence, it follows that $\lim\limits_{n\rightarrow\infty}q_n=0$, which gives us \eqref{PDEHlargen}.
\end{proof}

\begin{remark} From \eqref{dsrHR} and \eqref{dsr*HR*}, we have $R=-s_1\partial_{s_1}H$ and $R^{\star}=-2s_2\partial_{s_2}H$. Inserting them into \eqref{HRR*}, we arrive at \eqref{PDEHlargen}.
\end{remark}

\begin{remark}
In the case $s_2\rightarrow0^+$ {\rm(}so that $t_2\rightarrow0^+${\rm)} or $H$ is independent of $s_2$, equation \eqref{PDEHlargen} is reduced to
\begin{equation*}
\left(s_1
   H''\right){}^2+4\left(H'\right)^2 \left( s_1
   H'- H\right)-\left(\alpha H'+\frac{1}{2}\right)^2=0,
\end{equation*}
with $H=H(s_1,0)$, which agrees with (2.19) of \cite{ChenChen2015}. As pointed out in Remark 3 therein, the quantity $H(2s_1,0)-\frac{\alpha^2}{4}$ satisfies the $P_{III}$ equation of Ohyama-Kawamuko-Sakai-Okamoto. In addition, when $s_2\rightarrow0^+$, with the aid of \eqref{s20R*}, equation \eqref{HRR*} reduces to
\begin{equation*}
\begin{aligned}
H(s_1,0)=&\frac{1}{4}\left(\frac{s_1 }{R}R'
   -1\right
   )^2-R-\frac{1}{4}\left(\frac{s_1}{2R}-\alpha\right)^2,
\end{aligned}
\end{equation*}
where $R=R(s_1,0)$, which is consistent with (2.20) of \cite{ChenChen2015}. Note that our $s_1$ and $R$ correspond to $s$ and $sC(s)/2$ therein respectively.
\end{remark}

\subsection{The equilibrium density}
In this section, we consider the large $n$ behavior of our monic orthogonal polynomials $\{P_n(z;t_1,t_2)\}$
by using the linear statistics results \cite{ChenLawrence1998} which were derived via Dyson's Coulomb fluid method.
Now we give a brief description of this method.

Denote by $\{x_i\}_{i=1}^n$ the eigenvalues of the unitary ensemble associated with the weight function $w(x;t_1,t_2)$ given by \eqref{w}, i.e.
\[
w(x; t_1,t_2)=x^{\al}\:{\rm exp}\left(-x-\frac{t_1}{x}-\frac{t_2}{x^2}\right),\quad x\geq 0, ~\alpha>0, ~t_1>0,~t_2>0.
\]
The potential function reads
\begin{align}\label{pf}
v(x):=-\ln w(x;t_1,t_2)=-\al\:\ln\:x +x +\frac{t_1}{x}+\frac{t_2}{x^2}.
\end{align}
If we treat $\{x_i\}_{i=1}^n$ as $n$ identically charged particles, then when $n\rightarrow\infty$ they can be viewed as a continuous fluid with density $\sigma(x)$.
Here, the parameter ranges are narrowed from $\alpha>-1$ and $t_1\neq0$ to $\alpha>0$ and $t_1>0$ to ensure that $v''(x)=\frac{\alpha}{x^2}+\frac{2t_1}{x^3}+\frac{6t_2}{x^4}\geq0$ which indicates that $v(x)$ is convex and hence $\sigma(x)$ is supported on a single interval $[a,b]\subseteq[0,\infty)$.

In Dyson's works \cite{Dyson1962JMP}, $\sigma(x)$ is determined via the constrained minimization
\begin{gather*}
\min_{\sigma} F[\sigma] \qquad\text{subject to}\qquad \int_a^b \sigma(x)dx=n,
\end{gather*}
with
\begin{gather*}
F[\sigma]:=\int_a^b \sigma(x)v(x)dx-\int_a^b\int_a^b\sigma(x)\log|x-y|\sigma(y)dx dy.
\end{gather*}
According to Theorem 1.3 in Chapter I.1 of \cite{SaffTotik1997}, $\sigma(x)$ satisfies the following equilibrium condition
\begin{gather}\label{LA}
v(x)-2\int_a^b \log|x-y|\sigma(y)dy=A,\qquad x\in[a,b],
\end{gather}
where $A$ is the Lagrange multiplier which fixes $\int_a^b \sigma(x)dx=n$.
The derivative of this equation with respect to $x$ gives us
\begin{gather*}
2 \,\mathcal{P} \int_a^b \frac{\sigma(y)dy}{x-y}=v'(x),\qquad x\in[a,b],
\end{gather*}
where $\mathcal{P}$ denotes the Cauchy principal value.
Based on the theory of singular integral equations \cite{Mikhlin1957}, the solution of this equation subject to the boundary condition $\sigma(a)=0=\sigma(b)$ takes the form
\begin{gather}\label{sigmaint}
\sigma(x)=\frac{\sqrt{(b-x)(x-a)}}{2\pi^2}\int_a^b \frac{ v'(y)-v'(x)}{y-x}\frac{dy}{\sqrt{(b-y)(y-a)}},\qquad x\in(a,b),
\end{gather}
with a supplementary condition
\bea\label{v'cond1}
\int_{a}^{b}\frac{v'(x)dx}{\sqrtx}=0.
\eea
The normalization condition $\int_a^b \sigma(x)dx=n$ now reads
\bea\label{v'cond2}
\int_{a}^{b}\frac{x\:v'(x)dx}{\sqrtx}=2\pi n.
\eea

With the help of the integral identities in the Appendix, we obtain the expressions for $\sigma(x)$ and $A$.
\begin{proposition} The equilibrium density and the Lagrange multiplier are given by
\begin{equation}\label{sigma}
\frac{ 2\pi\sqrt{ab}\:\sigma(x) }{\sqrtx}=\frac{1}{x}\left(\al+\frac{a+b}{2ab}t_1+\frac{3(a+b)^2-4ab}{4(ab)^2}t_2\right)+\frac{1}{x^2}\left(t_1+\frac{a+b}{ab}t_2\right)+\frac{2t_2 }{x^3},
\end{equation}
for $x\in(a,b)$, and
\begin{align}\label{A}
A=\frac{a+b}{2}-\alpha\ln \frac{a+b+2\sqrt{ab}}{4}-2n\ln \frac{b-a}{4}+\frac{t_1}{\sqrt{ab}}+\frac{(a+b)t_2}{2(ab)^{3/2}},
\end{align}
where $a$ and $b$ are determined by the following two equalities
\begin{align}
\frac{\alpha}{\sqrt{ab}}+\frac{(a+b) t_1}{2(ab)^{3/2}}+\frac{3\left(\frac{a+b}{2}\right)^2-ab}{(ab)^{5/2}}t_2=1,\label{abeq1}\\
\frac{a+b}{2}-\frac{t_1}{\sqrt{ab}}-\frac{(a+b)t_2}{(ab)^{3/2}}=2n+\alpha.\label{abeq2}
\end{align}
\end{proposition}
\begin{proof}
From \eqref{pf}, we have
\begin{align*}
v'(x)=&-\frac{\alpha}{x} +1 -\frac{t_1}{x^2} -\frac{2t_2}{x^3},
\end{align*}
so that
\begin{align*}
\frac{v'(y)-v'(x)}{y-x}=\left(\frac{\alpha}{x}+\frac{t_1}{x^2}+\frac{2t_2}{x^3}\right)\frac{1}{y}+\left(\frac{t_1}{x}+\frac{2t_2}{x^2}\right)\frac{1}{y^2}+\frac{2t_2}{x}\cdot\frac{1}{y^3}.
\end{align*}
Inserting them into \eqref{sigmaint}-\eqref{v'cond2}, in view of the formulas in the Appendix, we are led to \eqref{sigma} and \eqref{abeq1}-\eqref{abeq2}.

Multiplying both sides of \eqref{LA} by $1/\sqrt{(b-x)(x-a)}$ and integrating them with respect to $x$ from $a$ to $b$, with the aid of the integrals in the Appendix and (5.18) of \cite{MuLyu24}, i.e.
\[\int_a^b\frac{\ln|x-y|dx}{\sqrt{(b-x)(x-a)}}=\pi\ln(b-a)-2\pi\ln2, \]
we arrive at \eqref{A}.
\end{proof}

We now apply \eqref{abeq1} and \eqref{abeq2}
to derive an algebraic equation for the geometric mean $\sqrt{ab}$.
\begin{lemma} The quantity
\[X:=\sqrt{ab}\]
satisfies an algebraic equation of degree nine
\begin{equation}\label{eqX}
\begin{aligned}
X^9&-\al\:X^8-((2n+\alpha)t_1+3t_2)\:X^6+(4\alpha\:t_2-t_1^2)\:X^5\\
&-3t_2(2n+\alpha)^2\: X^4-4t_1t_2(2n+\alpha)\:X^3-t_2(t_1^2+4\alpha\;t_2)\:X^2 + 4\:t_2^3=0,
\end{aligned}
\end{equation}
which, under the double scaling that $t_1\rightarrow0$, $t_2\rightarrow0^+$ and $n\rightarrow\infty$ such that $
s_1=2nt_1$ and $s_2=4n^2t_2$ are fixed,
reduces to
\begin{align}\label{eqX5}
X^5-\al\:X^4-s_1\:X^2-3\:s_2=0.
\end{align}
\end{lemma}
\begin{proof}
Let
$$Y:=\frac{a+b}{2}.$$
Then \eqref{abeq1} and \eqref{abeq2} become
\bean
-\frac{\alpha}{X}+1-\frac{t_1\:Y}{X^3}-\frac{3Y^2-X^2}{X^5}t_2=0,\\
2n+\al+\frac{t_1}{X}+\left(\frac{2t_2}{X^3}-1\right)Y=0.
\eean
Solving $Y$ from the second equation and inserting it into the first one leads us to \eqref{eqX}.

Replacing $t_1$ by $s_1/(2n)$ and $t_2$ by $s_2/(4n^2)$ in the left hand side of \eqref{eqX}, expanding it asymptotically for $n\rightarrow\infty$,
by setting the coefficient of the leading order term in $n$ to be 0, we get \eqref{eqX5}.
\end{proof}

\begin{remark} When $t_2\rightarrow0^+$, \eqref{sigma} reduces to
\begin{gather*}
\frac{ 2\pi\sqrt{ab}\:\sigma(x) }{\sqrtx}=\frac{1}{x}\left(\al+\frac{a+b}{2ab}t_1\right)+\frac{t_1 }{x^2},\qquad x\in(a,b),
\end{gather*}
which agrees with the equation below \emph{(2.5)} of \cite{ChenChen2015}, and \eqref{A} becomes
\begin{gather*}
A=\frac{a+b}{2}-\alpha\ln \frac{a+b+2\sqrt{ab}}{4}-2n\ln \frac{b-a}{4}+\frac{t_1}{\sqrt{ab}}.
\end{gather*}
Equation \eqref{eqX} simplifies to
\[
X^4-\al\:X^3-(2n+\alpha)t_1\:X-t_1^2=0,\]
which coincides with \emph{(2.8)} of \cite{ChenChen2015}. Sending $s_2$ to 0 in \eqref{eqX5} yields
\[
X^3-\alpha\:X^2-s_1=0,\]
which is consistent with the one given in part A of Section II in \cite{ChenChen2015}.
\end{remark}

\begin{remark}
Assuming $t_1\rightarrow0^+$ and $t_2\rightarrow0^+$, we get from \eqref{sigma}-\eqref{abeq2} that
\begin{align*}
\sigma(x)=&\frac{\al}{2\pi\sqrt{ab}}\cdot\frac{\sqrtx}{x},\qquad x\in(a,b),\\
A=&\frac{a+b}{2}-\alpha\ln \frac{a+b+2\sqrt{ab}}{4}-n\ln \frac{(b-a)^2}{16},
\end{align*}
and
\[\sqrt{ab}=\alpha,\qquad\qquad \frac{a+b}{2}=2n+\alpha.\]
Substituting the values of $\sqrt{ab}$ and $\frac{a+b}{2}$ into $\sigma(x)$ and $A$, and noting that $(b-a)^2=(a+b)^2-4ab$, we obtain
\begin{align*}
\sigma(x)=&\frac{1}{2\pi}\cdot\frac{\sqrt{-x^2+2(2n+\alpha)x-\alpha^2}}{x},\qquad x\in(a,b),\\
A=&2n+\alpha-(n+\alpha)\ln (n+\alpha)-n\ln n.
\end{align*}
Setting $x=4ny$ in the above expression of $\sigma(x)$ and supposing $\alpha$ is finite, by sending $n\rightarrow\infty$, we find
\[\lim\limits_{n\rightarrow\infty}\sigma(4ny)=\frac{1}{2\pi}\sqrt{\frac{1-y}{y}},\qquad y\in(0,1),\]
which is the Marchenko-Pastur density \cite[(19.1.11)]{Mehta2006}. For the derivation of $y\in(0,1)$, see Remark 11 of \cite{MuLyu24}.
\end{remark}

		\section{Laguerre Weight with Three Time Variables}
		We consider the following weight function
		\begin{equation}
			w(x;t_1,t_2,t_3)=x^\alpha \exp\left(-x-\frac{t_1}{x}-\frac{t_2}{x^2}-\frac{t_3}{x^3}\right),\quad x\in[0,+\infty),~\alpha>-1, \end{equation}
with $t_1,t_2\in\mathbb{R}\setminus\{0\}$ and $t_3>0$. We adopt the same notation $D_n, P_n, h_n, \alpha_n,\beta_n, p(n)$ as in the two-variable case, with all the quantities now functions of $t_1,t_2,t_3$ instead of just $t_1$ and $t_2$. The properties involving them, as stated in the Introduction, remain valid. In this section, we omit explicit dependence on $t_1,t_2,t_3$ unless necessary.

		We have
		\[v(x)=v(x;t_1,t_2,t_3):=-\ln w(x;t_1,t_2,t_3)=-\alpha\ln x+x+\frac{t_1}{x}+\frac{t_2}{x^2}+\frac{t_3}{x^3},\]
		and
		\begin{equation}	 v'(x)=-\frac{\alpha}{x}+1-\frac{t_1}{x^2}-\frac{2t_2}{x^3}-\frac{3t_3}{x^4}.
		\end{equation}
		It follows that
		\begin{equation} \frac{zv'(z)-yv'(y)}{z-y}=1+\frac{1}{z}\left(\frac{t_1}{y}+\frac{2t_2}{y^2}+\frac{3t_3}{y^3}\right)+\frac{1}{z^2}\left(\frac{2t_2}{y}+\frac{3t_3}{y^2}\right)+\frac{1}{z^3}\cdot\frac{3t_3}{y}.
		\end{equation}
		Substituting it into \eqref{defAn}-\eqref{defBn} which are valid for $w(x;t_1,t_2,t_3)$, we get the expressions for $A_n$ and $B_n$.
		\begin{lemma} $A_n$ and $B_n$ are given by
			\begin{align}
				 A_n(z)&=\frac{1}{z}+\frac{R_n+R_n^{\star}+\hat{R}_n}{z^2}+\frac{\tau R_n+\rho R_n^{\star}}{z^3}+\frac{\tau \rho R_n}{z^4},\label{An-3}\\		 B_n(z)&=\frac{-n}{z}+\frac{r_n+r_n^{\star}+\hat{r}_n}{z^2}+\frac{\tau r_n+\rho r_n^{\star}}{z^3}+\frac{\tau \rho\, r_n}{z^4},\label{Bn-3}
			\end{align}
	where
\[
\tau:=\frac{2t_2}{t_1}, \qquad\quad\rho:=\frac{3t_3}{2t_2}, \qquad\quad\tau\rho=\frac{3t_3}{t_1},\]
 and
			\begin{align*}
				 R_n&:=\frac{t_1}{h_n}\int_{0}^{+\infty}P_n^2(y)w(y)\frac{dy}{y},&r_n&:=\frac{t_1}{h_{n-1}}\int_{0}^{+\infty}P_n(y)P_{n-1}(y)w(y)\frac{dy}{y},\\
				 R_n^{\star}&:=\frac{2t_2}{h_n}\int_{0}^{+\infty}P_n^2(y)w(y)\frac{dy}{y^2},&r_n^{\star}&:=\frac{2t_2}{h_{n-1}}\int_{0}^{+\infty}P_n(y)P_{n-1}(y)w(y)\frac{dy}{y^2},\\
				 \hat{R}_n&:=\frac{3t_3}{h_n}\int_{0}^{+\infty}P_n^2(y)w(y)\frac{dy}{y^3},&\hat{r}_n&:=\frac{3t_3}{h_{n-1}}\int_{0}^{+\infty}P_n(y)P_{n-1}(y)w(y)\frac{dy}{y^3}.
			\end{align*}	
		\end{lemma}
		
		\subsection{Difference Equations for the Auxiliary Quantities}	
		Substituting \eqref{An-3}-\eqref{Bn-3} into $(S_1)$ and equating the coefficients of $z^{-j}$ (for $j=1,2,3,4$) yields
		\begin{gather}		 \alpha_n=2n+1+\alpha+R_n+R_n^{\star}+\hat{R}_n,\label{al-3}\\
r_{n+1}+r_n+r_{n+1}^{\star}+r_n^{\star}+\hat{r}_{n+1}+\hat{r}_n=t_1-\alpha_n(R_n+R_n^{\star}+\hat{R}_n)+\tau R_n+\rho R_n^{\star},\label{hatrn+1hattnt1t2t3}\\
\tau (r_{n+1}+r_n)+\rho (r_{n+1}^{\star}+r_n^{\star})=2t_2-\alpha_n(\tau R_n+\rho R_n^{\star})+\tau \rho R_n,\label{rn+1*rn*=2t23t3}\\
\tau \rho (r_{n+1}+r_n)=3t_3-\alpha_n\tau \rho R_n.\label{3t3rn+1rn=1-anRn}
		\end{gather}
		Since $\tau \rho=3t_3/t_1\neq0$, we divide both sides of \eqref{3t3rn+1rn=1-anRn} by $\tau \rho $ and get
		\begin{equation}
			\label{rn+1rn}
			r_{n+1}+r_n=t_1-\alpha_nR_n.	
		\end{equation}
		Inserting it into \eqref{rn+1*rn*=2t23t3} and dividing the resulting equation by $\rho$ gives
		\begin{equation}
			\label{rn+1*rn*}
			r_{n+1}^{\star}+r_n^{\star}=\tau R_n-\alpha_nR_n^{\star}.	
		\end{equation}
Plugging \eqref{rn+1rn} and \eqref{rn+1*rn*} into \eqref{hatrn+1hattnt1t2t3} results in
		\begin{equation}
			\label{hatrn+1hatrn}
			\hat{r}_{n+1}+\hat{r}_n=\rho R_n^{\star}-\alpha_n\hat{R}_n.	
		\end{equation}

Substituting \eqref{An-3}-\eqref{Bn-3} into $(S_2')$ and comparing the coefficients of $z^{-j}$ (for $j=2,3,\dots,8$) on both sides, we obtain
		\begin{equation}
			\label{bt-3}	 \beta_n=n(n+\alpha)+r_n+r_n^{\star}+\hat{r}_n+\sum_{j=0}^{n-1}\left(R_j+R_j^{\star}+\hat{R}_j\right),
		\end{equation}
\begin{equation*}
\begin{aligned}		 \beta_n&(R_n+R_{n-1}+R_n^{\star}+R_{n-1}^{\star}+\hat{R}_n+\hat{R}_{n-1})\\
&=-(2n+\alpha)(r_n+r_n^{\star}+\hat{r}_n)+nt_1+\tau r_n+\rho r_n^{\star}+\tau \sum_{j=0}^{n-1}R_j+\rho \sum_{j=0}^{n-1}R_j^{\star},
\end{aligned}
\end{equation*}
\begin{equation*}
\begin{aligned}		 &\beta_n\left[(R_n+R_n^{\star}+\hat{R}_n)(R_{n-1}+R_{n-1}^{\star}+\hat{R}_{n-1})+\tau  (R_n+R_{n-1})+\rho (R_n^{\star}+R_{n-1}^{\star})\right]\\			 &=(r_n+r_n^{\star}+\hat{r}_n)^2-t_1(r_n+r_n^{\star}+\hat{r}_n)-(2n+\alpha)(\tau r_n+\rho r_n^{\star})+2nt_2+\tau \rho r_n+\tau \rho \sum_{j=0}^{n-1}R_j,
\end{aligned}
\end{equation*}
\begin{equation}\label{z-5}
\begin{aligned}		
\beta_n&\left[\tau \rho (R_n+R_{n-1})+(\tau R_n+\rho R_n^{\star})(R_{n-1}+R_{n-1}^{\star}+\hat{R}_{n-1})+(\tau R_{n-1}+\rho R_{n-1}^{\star})(R_n+R_n^{\star}+\hat{R}_n)\right]\\
&=3n t_3-(2n+\alpha)\tau \rho r_n+2(r_n+r_n^{\star}+\hat{r}_n)(\tau r_n+\rho r_n^{\star})-2t_2(r_n+r_n^{\star}+\hat{r}_n)-t_1(\tau r_n+\rho r_n^{\star}),
\end{aligned}
\end{equation}
\begin{equation}\label{z-6}	
\begin{aligned}
\tau \rho \beta_n&\left[R_n(R_{n-1}+R_{n-1}^{\star}+\hat{R}_{n-1})+R_{n-1}(R_n+R_n^{\star}+\hat{R}_n)\right]+\beta_n(\tau R_n+\rho R_n^{\star})(\tau R_{n-1}+\rho R_{n-1}^{\star})\\	    &=(r_n+r_n^{\star}+\hat{r}_n)(2\tau \rho r_n-3t_3)-3t_3r_n+(\tau r_n+\rho r_n^{\star})^2-2t_2(\tau r_n+\rho r_n^{\star}),
\end{aligned}
\end{equation}
		\begin{align}\label{z-7}			
\tau \rho \beta_n\left[R_n(\tau R_{n-1}+\rho R_{n-1}^{\star})+R_{n-1}(\tau R_n+\rho R_n^{\star})\right]=(2\tau \rho r_n-3t_3)(\tau r_n+\rho r_n^{\star})-3t_3\tau r_n,
		\end{align}
		\begin{equation}
			\label{z-8}
			(\tau \rho )^2\beta_nR_nR_{n-1}=(\tau \rho)^2r_n(r_n-t_1).
		\end{equation}
		Since $\tau \rho=3t_3/t_1\neq0$, we divide \eqref{z-7}  by $\tau \rho $ and \eqref{z-8} by $(\tau \rho)^2$ to get
		\begin{gather}
			\beta_n\left[R_n(\tau R_{n-1}+\rho R_{n-1}^{\star})+R_{n-1}(\tau R_n+\rho R_n^{\star})\right]=(2r_n-t_1)(\tau r_n+\rho r_n^{\star})-2t_2r_n,\nonumber\\
\beta_nR_nR_{n-1}=r_n(r_n-t_1).\label{bnRnRn-1}
		\end{gather}
		Plugging the second equation into the first one produces
		\begin{equation}
			\label{bnRnRn-1*bnRn-1Rn*}
			 \beta_n(R_nR_{n-1}^{\star}+R_{n-1}R_n^{\star})=r_n^{\star}(2r_n-t_1).
		\end{equation}
Inserting \eqref{bnRnRn-1} and \eqref{bnRnRn-1*bnRn-1Rn*} into \eqref{z-6}, and dividing the obtained equation by $\rho$, we come to		 \begin{equation}
			\label{RnhatRn-1hatRnRn-1Rn*Rn-1^*}
			\tau \beta_n(R_n\hat{R}_{n-1}+\hat{R}_nR_{n-1})+\rho \beta_n R_n^{\star}R_{n-1}^{\star}=\rho(r_n^{\star})^2+2\hat{r}_n(\tau r_n-t_2).
		\end{equation}
		Substituting \eqref{bnRnRn-1}-\eqref{RnhatRn-1hatRnRn-1Rn*Rn-1^*} into \eqref{z-5} and then dividing it by $\rho$ leads to
		\begin{align}\label{6t3hatrnrn*}
			\beta_n\left[\tau (R_n+R_{n-1})+R_{n-1}^{\star}\left(R_n^{\star}+\hat{R}_n\right)+R_n^{\star}\hat{R}_{n-1}\right]=(r_n^{\star})^2-\tau (2n+\alpha)r_n+2n t_2+2\hat{r}_nr_n^{\star}.
		\end{align}

Combining \eqref{bnRnRn-1} with \eqref{bnRnRn-1*bnRn-1Rn*} to eliminate $R_{n-1}$ yields
			\begin{equation}\label{bnsRn-1} \beta_nR_{n-1}^{\star}=\frac{r_n^{\star}(2r_n-t_1)}{R_n}-\frac{r_n(r_n-t_1)R_n^{\star}}{R_n^2}.
			\end{equation}
Using it and  \eqref{bnRnRn-1} to eliminate $R_{n-1}^{\star}$ and $R_{n-1}$ respectively in \eqref{RnhatRn-1hatRnRn-1Rn*Rn-1^*} gives us
			\begin{equation}\label{Rh-3}		 	 \beta_n\hat{R}_{n-1}=\frac{1}{R_n}\left(\hat{r}_n(2r_n-t_1)+r_n(t_1-r_n)\frac{\hat{R}_n}{R_n}\right)+\frac{\rho}{\tau R_n}\left(r_n^{\star}-r_n\frac{R_n^{\star}}{R_n} \right)\left(r_n^{\star}+(t_1-r_n)\frac{R_n^{\star}}{R_n} \right).
	\end{equation}
			Plugging it, \eqref{bnsRn-1} and \eqref{bnRnRn-1} into \eqref{6t3hatrnrn*}, we express $\beta_n$ in terms of the auxiliary quantities. The expression together with \eqref{al-3} is stated in the following lemma.			 
		\begin{lemma} The recurrence coefficients are expressed in terms of the auxiliary quantities as
			\begin{equation}
				\label{an}
				 \alpha_n=2n+1+\alpha+R_n+R_n^{\star}+\hat{R}_n,
			\end{equation}
			\begin{equation}
\begin{aligned}\label{bn}		
\beta_n&=\frac{1}{\tau R_n}\left(1-\frac{\rho R_n^{\star}}{\tau R_n}\right)\left(r_n^{\star}-r_n\frac{R_n^{\star}}{R_n} \right)\left(r_n^{\star}+(t_1-r_n)\frac{R_n^{\star}}{R_n} \right)+\frac{2\hat{r}_nr_n^{\star}}{\tau R_n}\\
&+\frac{r_n(t_1-r_n)}{R_n^2}\left(1-2\frac{R_n^{\star}\hat{R}_n}{\tau R_n}\right)+\frac{t_1-2 r_n}{\tau R_n^2}\left( \hat{r}_n R_n^{\star}+\hat{R}_n r_n^{\star}\right)+\frac{nt_1-(2n+\alpha)r_n}{R_n}.
	\end{aligned}
	\end{equation}
		\end{lemma}

Moreover, the auxiliary quantities satisfy a system of first order difference equations that can be iterated in $n$.
	\begin{proposition}
			The auxiliary quantities satisfy the following difference equations
\begin{align}
 r_{n+1}=&-r_n+t_1-(2n+1+\alpha+R_n+R_n^{\star}+\hat{R}_n)R_n,\label{rn+1rnRn}\\
r_{n+1}^{\star}=&-r_n^{\star}+\tau R_n-(2n+1+\alpha+R_n+R_n^{\star}+\hat{R}_n)R_n^{\star}, \label{rn+1*rn*Rn}\\
\hat{r}_{n+1}=&-\hat{r}_n+\rho R_n^{\star}-(2n+1+\alpha+R_n+R_n^{\star}+\hat{R}_n)\hat{R}_n,\label{hatrn+1hatrnRn}
\end{align}
for $n\geq0$, and
\begin{equation}
\begin{aligned}
	\label{difequ5}
	 R_n&\left[(r_n^{\star}R_{n-1}-r_nR_{n-1}^{\star})(r_n^{\star}R_{n-1}-(r_n-t_1)R_{n-1}^{\star})\left(\frac{\rho}{\tau}R_{n-1}^{\star}-R_{n-1}\right)\right.\\&\left.+(2r_n-t_1)(\hat{r}_nR_{n-1}^{\star}+r_n^{\star}\hat{R}_{n-1})R_{n-1}^2+ r_n(r_n-t_1)(\tau R_{n-1}-2R_{n-1}^{\star}\hat{R}_{n-1})R_{n-1}\right.\\&\left.+( (2n+\alpha)\tau r_n-2nt_2-2\hat{r}_nr_n^{\star})R_{n-1}^3\right]=\tau r_n(t_1-r_n)R_{n-1}^3,
\end{aligned}
\end{equation}
\begin{equation}\label{difequ6}
	 r_n(r_n-t_1)R_{n-1}R_n^{\star}=[r_n^{\star}(2r_n-t_1)R_{n-1}+r_n(t_1-r_n)R_{n-1}^{\star}]R_n,
\end{equation}
\begin{equation}\label{difequ4}
\begin{aligned}
	\tau r_n(r_n-t_1)R_{n-1}^2\hat{R}_n=&R_n\left[\rho(r_n^{\star}R_{n-1}-r_nR_{n-1}^{\star})(r_n^{\star}R_{n-1}+(t_1-r_n)R_{n-1}^{\star})\right. \\ &\qquad\left.+\tau R_{n-1} \left(r_n(t_1-r_n)\hat{R}_{n-1}+(2r_n-t_1)\hat{r}_nR_{n-1}\right) \right],
\end{aligned}
\end{equation}	
for $n\geq1$,
with the initial conditions given by
\begin{gather*}
r_{0}=r_{0}^{\star}=\hat{r}_0=0,
\end{gather*}
and
\begin{gather*}
R_{0}=t_1\frac{\int_{0}^{+\infty}w(y;t_1,t_2,t_3)\frac{dy}{y}}{\int_{0}^{+\infty}w(y;t_1,t_2,t_3)dy},
~~ R_{0}^{\star}=2t_2\frac{\int_{0}^{+\infty}w(y;t_1,t_2,t_3)\frac{dy}{y^2}}{\int_{0}^{+\infty}w(y;t_1,t_2,t_3)dy},~~ \hat{R}_{0}=3t_3\frac{\int_{0}^{+\infty}w(y;t_1,t_2,t_3)\frac{dy}{y^3}}{\int_{0}^{+\infty}w(y;t_1,t_2,t_3)dy}.
\end{gather*}
		\end{proposition}

\begin{proof}
			Plugging \eqref{al-3} into \eqref{rn+1rn}-\eqref{hatrn+1hatrn} yields \eqref{rn+1rnRn}-\eqref{hatrn+1hatrnRn}.
Using \eqref{bnRnRn-1} to eliminate $\beta_n$ in \eqref{bnRnRn-1*bnRn-1Rn*}-\eqref{6t3hatrnrn*}, we come to \eqref{difequ6}, and
\begin{align}\label{difeq7}
				r_n(r_n-t_1)(\tau (R_n\hat{R}_{n-1}+\hat{R}_nR_{n-1})+\rho R_n^{\star}R_{n-1}^{\star})=R_nR_{n-1}[\rho (r_n^{\star})^2+2\hat{r}_n(\tau r_n-t_2)],
			\end{align}
			\begin{equation}\label{difeq8}
			\begin{aligned}
				r_n(r_n-t_1)[\tau (R_n+R_{n-1})+&R_n^{\star}(R_{n-1}^{\star}+\hat{R}_{n-1})+R_{n-1}^{\star}\hat{R}_n]\\
				&=R_nR_{n-1}[(r_n^{\star})^2-\tau (2n+\alpha)r_n+2nt_2+2\hat{r}_nr_n^{\star}].
			\end{aligned}
			\end{equation}
			Eliminating $R_n^{\star}$ in \eqref{difeq7} using \eqref{difequ6} leads to \eqref{difequ4}. Using it and \eqref{difequ6} to eliminate $\hat{R}_n$ and $R_n^{\star}$ respectively in \eqref{difeq8},  we arrive at \eqref{difequ5}.		
		\end{proof}
\begin{remark}
Setting $n=0$ in \eqref{rn+1rnRn}-\eqref{hatrn+1hatrnRn} produces the values of $r_1,r_1^{\star}$ and $\hat{r}_1$. Inserting them into \eqref{difequ5} with $n=1$ gives us $R_1$. Using \eqref{difequ6} and \eqref{difequ4}, both with $n=1$, we find $R_1^{\star}$ and $\hat{R}_1$.
\end{remark}
\begin{remark}
			When $t_3\to0^+$, we have $\hat{R}_n\rightarrow0$ and $\hat{r}_n\rightarrow0$. Equations \eqref{an}-\eqref{rn+1*rn*Rn} and \eqref{difequ5}-\eqref{difequ6} are reduced to \eqref{alphanRn}-\eqref{Rnrndiffeq4}.
		\end{remark}

At the end of this section, we provide two expressions for $p(n,t_1,t_2,t_3)$ (i.e. the coefficient of $x^{n-1}$ in $P_n(x;t_1,t_2,t_3)$), which will be used in the next section. Specifically, from \eqref{al-3}, it follows that
	\begin{equation*}		 \sum_{j=0}^{n-1}\alpha_j=n(n+\alpha)+\sum_{j=0}^{n-1}\left(R_j+R_j^{\star}+\hat{R}_j\right).
	\end{equation*}
Combining it with  \eqref{bt-3}, in view of \eqref{a13}, we come to the desired results.	
\begin{lemma}
			We have
			\begin{align} p(n)&=-n(n+\alpha)-\sum_{j=0}^{n-1}\left(R_j+R_j^{\star}+\hat{R}_j\right)	 \label{pn-3}\\				 &=r_n+r_n^{\star}+\hat{r}_n-\beta_n.\label{p(n)3_2}
			\end{align}

		\end{lemma}		
		
		\subsection{Toda equation}
		Differentiating both sides of the following two orthogonality identities with respect to $t_1, t_2, t_3$,
		\begin{gather}
			 h_n(t_1,t_2,t_3)=\int_{0}^{+\infty}P_n^2(x;t_1,t_2,t_3)w(x;t_1,t_2,t_3)dx,\\
			 0=\int_{0}^{+\infty}P_n(x;t_1,t_2,t_3)P_{n-1}(x;t_1,t_2,t_3)w(x;t_1,t_2,t_3)dx,
		\end{gather}
 we obtain the differential relations given below.		
		\begin{lemma} \begin{enumerate}[\upshape(i)]
				\item The derivatives of $h_n$ are linked with $\{R_n,R_n^{\star}, \hat{R}_n\}$ by
				\begin{align}\label{Dh-3}
					t_1\partial_{t_1} \ln h_n=-R_n,\qquad
					2t_2\partial_{t_2} \ln h_n=&-R_n^{\star},\qquad
					3t_3\partial_{t_3} \ln h_n=-\hat{R}_n.
				\end{align}
				In view of $\beta_n=h_n/h_{n-1}$, we have
				\begin{align}\label{Dbt-3}
					t_1\partial_{t_1} \ln \beta_n=R_{n-1}-R_n,\quad
					2t_2\partial_{t_2} \ln \beta_n=R_{n-1}^{\star}-R_n^{\star},\quad
					3t_3\partial_{t_3} \ln \beta_n=\hat{R}_{n-1}-\hat{R}_n.
				\end{align}
				\item The derivatives of $p(n,t_1,t_2,t_3)$ are related to $\{r_n,r_n^{\star}, \hat{r}_n\}$ by
				\begin{align}\label{Dp-3}
					t_1\partial_{t_1} p(n)=r_n,\qquad\quad
					2t_2\partial_{t_2} p(n)=r_n^{\star},\qquad\quad
					3t_3\partial_{t_3} p(n)=\hat{r}_n.
				\end{align}
				Since $\alpha_n=p(n)-p(n+1)$, it follows that
				\begin{align}\label{Dal-3}
					t_1\partial_{t_1} \alpha_n=&r_n-r_{n+1},\qquad					 2t_2\partial_{t_2}\alpha_n=r_n^{\star}-r_{n+1}^{\star},\qquad					 3t_3\partial_{t_3}\alpha_n=\hat{r}_n-\hat{r}_{n+1}.
				\end{align}	
			\end{enumerate}
		\end{lemma}
Adding the three identities in \eqref{Dbt-3} and \eqref{Dal-3} separately, in view of \eqref{al-3} and \eqref{bt-3},  we arrive at the Toda-like equations for $\alpha_n$ and $\beta_n$.
		\begin{proposition}
			The recurrence coefficients satisfy the following partial differential-difference equations
			\begin{gather}
				 (t_1\partial_{t_1}+2t_2\partial_{t_2}+3t_3\partial_{t_3})\alpha_n=\beta_n-\beta_{n+1}+\alpha_n,\\
				 (t_1\partial_{t_1}+2t_2\partial_{t_2}+3t_3\partial_{t_3})\ln\beta_n=\alpha_{n-1}-\alpha_n+2,
			\end{gather}
			from which follow a second order Toda-like  equation for $\beta_n$
			\begin{equation}
			\begin{aligned}			 (t_1^2\partial_{t_1t_1}+2t_1t_2(\partial_{t_1t_2}+\partial_{t_2t_1})+4t_2^2\partial_{t_2t_2}+2t_2\partial_{t_2}+9t_3^2\partial_{t_3t_3}+3t_1t_3(\partial_{t_1t_3}+\partial_{t_3t_1})\\
+6t_2t_3(\partial_{t_2t_3}+\partial_{t_3t_2})+6t_3\partial_{t_3})\ln\beta_n=\beta_{n-1}-2\beta_n+\beta_{n+1}-2.
			\end{aligned}
			\end{equation}
		\end{proposition}
\subsection{Riccati-like Equations and PDEs}
\begin{lemma}
		The auxiliary quantities $\{R_n, R_n^{\star}, \hat{R}_n,r_n, r_n^{\star}, \hat{r}_n\}$ satisfy the following PDE system
		\begin{gather}
t_1\partial_{t_1}(R_n+R_n^{\star}+\hat{R}_n)=2r_n+(2n+1+\alpha+R_n+R_n^{\star}+\hat{R}_n)R_n-t_1,\label{t1RnRic}\\
2t_2\partial_{t_2}(R_n+R_n^{\star}+\hat{R}_n)=2r_n^{\star}+(2n+1+\alpha+R_n+R_n^{\star}+\hat{R}_n)R_n^{\star}-\tau R_n,\label{2t2RnRic}\\	
3t_3\partial_{t_3}(R_n+R_n^{\star}+\hat{R}_n)=2\hat{r}_n+(2n+1+\alpha+R_n+R_n^{\star}+\hat{R}_n)\hat{R}_n-\rho R_n^{\star},\label{3t3RnRic}\\
t_1\partial_{t_1}(r_n+r_n^{\star}+\hat{r}_n)=\Xi_n+r_n+\frac{2r_n(r_n-t_1)}{R_n},\label{t1rnric}\\
2t_2\partial_{t_2}(r_n+r_n^{\star}+\hat{r}_n)=\frac{R_n^{\star}}{R_n}\Xi_n+r_n^{\star}+\frac{r_n^{\star}(2r_n-t_1)}{R_n},\label{2t2rnric}\\
3t_3\partial_{t_3}(r_n+r_n^{\star}+\hat{r}_n)=\frac{\hat{R}_n}{R_n}\Xi_n+\hat{r}_n +\frac{\hat{r}_n(2r_n-t_1)}{R_n} +\kappa_n,\label{3t3rnric}
		\end{gather}
		where
\begin{align*}			 \Xi_n:=&\kappa_n\left(\frac{R_n^{\star}}{\tau}-\frac{R_n}{\rho}\right)+2r_n(t_1-r_n)\frac{R_n^{\star}\hat{R}_n}{\tau R_n^2}+\frac{2r_n-t_1}{\tau R_n}\left(\hat{r}_nR_n^{\star}+\hat{R}_nr_n^{\star} \right)\\
&-\frac{2\hat{r}_nr_n^{\star}}{\tau}+(2n+\alpha)r_n-nt_1,\\
\kappa_n:=&\frac{\rho}{\tau}\left(\frac{r_n^{\star}}{R_n}-\frac{r_nR_n^{\star}}{R_n^2} \right)\left(r_n^{\star}-\frac{(r_n-t_1)R_n^{\star}}{R_n} \right).
\end{align*}
\end{lemma}

\begin{remark}
Solving $r_n,r_n^{\star}$ and $\hat{r}_n$ from \eqref{t1RnRic}-\eqref{3t3RnRic}, and substituting the obtained expressions into \eqref{t1rnric}-\eqref{3t3rnric}, we come to the system of three PDEs satisfied by $R_n,R_n^{\star}$ and $\hat{R}_n$. Since the equations are too complicated, we do not display them here.
\end{remark}

Define
		\begin{equation}
			H_n(t_1,t_2,t_3):=(t_1 \partial_{t_1}+2t_2 \partial_{t_2}+3t_3 \partial_{t_3})\ln D_n(t_1,t_2,t_3).
		\end{equation}
Recalling that $D_n=\prod_{j=0}^{n-1}h_j$, by using \eqref{Dh-3} and \eqref{pn-3}, we find	
\begin{equation}\label{3Hnpn}				 H_n(t_1,t_2,t_3)=n(n+\alpha)+p(n,t_1,t_2),
\end{equation}
which, in view of \eqref{Dp-3}, gives us			
\begin{equation}\label{DHn-3}
t_1 \partial_{t_1}H_n=r_n,\qquad\qquad
2t_2 \partial_{2t_2}H_n=r_n^{\star},\qquad\qquad
3t_3 \partial_{3t_3}H_n=\hat{r}_n.
			\end{equation}
		
Inserting  \eqref{p(n)3_2} into \eqref{3Hnpn} results in
		\begin{align}
			 \beta_{n}&=r_n+r_n^{\star}+\hat{r}_n-H_n+n(n+\alpha)\label{3bnHn_1}\\
			&=(t_1 \partial_{t_1}+2t_2 \partial_{t_2}+3t_3 \partial_{t_3})H_n-H_n+n(n+\alpha),\label{3bnHn_2}
		\end{align}
		where the second identity is due to \eqref{DHn-3}.
		Combining \eqref{3bnHn_1} with  \eqref{bn} leads to an expression for $H_n$ in terms of the auxiliary quantities. Replacing $r_n$, $r_n^{\star}$ and $\hat{r}_n$ in this expression  using \eqref{t1RnRic}-\eqref{3t3RnRic}, we express $H_n$ in terms of $\{R_n,R_n^{\star},\hat{R}_n\}$ and their first order derivatives.		
		
Now we turn to finding the expressions for $\{R_n,R_n^{\star},\hat{R}_n\}$ in terms of $H_n$ and its derivatives, and finally derive the PDE for $H_n$. Using \eqref{bnRnRn-1}, \eqref{bnsRn-1} and \eqref{Rh-3} to get rid of $R_{n-1}, R_{n-1}^{\star}$ and $\hat{R}_{n-1}$ in \eqref{Dbt-3} yields
		\begin{gather}
			t_1 \partial_{t_1} \beta_n=\frac{r_n(r_n-t_1)}{R_n}-\beta_{n}R_n,\label{bnrnRn}\\
			2t_2 \partial_{t_2}\beta_n=\frac{r_n^{\star}(2r_n-t_1)}{R_n}+\frac{r_n(t_1-r_n)R_n^{\star}}{R_n^2}-\beta_nR_n^{\star},\label{bnrn*Rn*}\\
			3t_3 \partial_{t_3}\beta_n=\frac{\hat{r}_n(2r_n-t_1)}{R_n}+\frac{r_n(t_1-r_n)\hat{R}_n}{R_n^2}+\frac{\rho}{\tau R_n}\left(r_n^{\star}-\frac{r_nR_n^{\star}}{R_n} \right)\left(r_n^{\star}+\frac{(t_1-r_n)R_n^{\star}}{R_n} \right)-\beta_n\hat{R}_n.\label{bnhatrnRn}
		\end{gather}
Noting that equations \eqref{bnrnRn}-\eqref{bnrn*Rn*} and the first two identities in \eqref{Dbt-3} are identical in form to \eqref{eq2Rn}-\eqref{eq2Rn*} and \eqref{lnbetaD1} respectively, we know that the expressions for $R_n$ and $R_n^{\star}$ given by \eqref{Rnsol} and  \eqref{Rn*sol} still hold.
		Solving $\hat{R}_n$ from \eqref{bnhatrnRn} yields
		\begin{equation}\label{3hatRnSolv}
			 \hat{R}_n=\frac{\hat{r}_n(2r_n-t_1)+\frac{\rho}{\tau}\left(r_n^{\star}-\frac{r_nR_n^{\star}}{R_n} \right)\left(r_n^{\star}-\frac{(r_n-t_1)R_n^{\star}}{R_n} \right)-\left(3t_3 \partial_{t_3}\beta_n \right)R_n}{\frac{r_n(r_n-t_1)}{R_n}+\beta_nR_n}.
		\end{equation}
Replacing the factor $\frac{r_n(r_n-t_1)}{R_n}$ in the denominator using \eqref{bnrnRn}, substituting \eqref{Rnsol} and  \eqref{Rn*sol} into the resulting expression to get rid of $R_n$ and $R_n^{\star}$, then replacing $\{r_n,r_n^{\star},\hat{r}_n\}$ using \eqref{DHn-3}, we arrive at the PDE for $H_n$. Since the equation is complicated, we omit its explicit form.

\begin{remark}\label{R1-3}
Assuming $t_1\rightarrow0, t_2\rightarrow0, t_3\rightarrow0^+$ and $n\rightarrow\infty$ such that $s_1=2nt_1,s_2=4n^2t_2$ and $ s_3=8n^3t_3$ are fixed, via an argument similar to the one used in Section \ref{dsRH}, we obtain the limiting equations for the double scaled $\{R_n, R_n^{\star}, \hat{R}_n\}$  and $H_n$
from their finite-dimensional ones.
\end{remark}

		\section{Generalization to the Laguerre Weight with $m$ Time Variables}
		Now we consider the weight function
		\begin{align}\label{wm}
			w(x; \vec{t}\:)=x^{\al}{\rm exp}\left(-x-\sum_{k=1}^m\frac{t_k}{x^k}\right),\quad x\in[0,+\infty), ~\alpha>-1,
		\end{align}
where $\vec{t}=(t_1,t_2,\dots,t_m)$ with $t_m>0$ and $t_i\in\mathbb{R}\setminus\{0\}$ for $i=1,2,\dots,m-1$. The factor $\sum_{k=1}^mt_k/x^k$ with $t_m>0$ makes $w(x;\vec{t}\:)$ vanish infinitely fast as $x\rightarrow0^+$. We use the same notation $D_n, P_n, h_n, \alpha_n,\beta_n, p(n)$ as in the $m=2$ case, but now they depend on $\vec{t}$. The properties of the monic orthogonal polynomials, presented in the Introduction, still hold.
		
	We have
			\begin{align*}
				v(x)&=-\ln w(x;\vec{t}\:)=-\al\:\ln\:x +x +\sum_{k=1}^m\frac{t_k}{x^k},\\
				v'(x)&=-\frac{\alpha}{x} +1 -\sum_{k=1}^m\frac{kt_k}{x^{k+1}},
			\end{align*}
so that
			\begin{equation*}
	 zv'(z)-yv'(y)=z-y+\sum_{k=1}^{m}\frac{kt_k(z^k-y^k)}{z^ky^k}.
			\end{equation*}
Since
			\begin{equation*}
				 z^k-y^k=(z-y)\sum_{j=1}^{k}z^{k-j}y^{j-1},
			\end{equation*}
		it follows that
\begin{align*}				 \frac{zv'(z)-yv'(y)}{z-y}&=1+\sum_{k=1}^{m}\left(kt_k\sum_{j=1}^{k}\frac{1}{z^j\cdot y^{k+1-j}}\right)\\
				 &=1+\sum_{\ell=1}^{m}\frac{1}{z^{\ell}}\cdot\left( \sum_{i=1}^{m+1-\ell} \frac{(\ell-1+i)t_{\ell-1+i}}{y^i} \right).
\end{align*}
			Substituting the second identity into \eqref{defAn}-\eqref{defBn}, we obtain the expressions for  $A_n$ and $B_n$.
		\begin{lemma}
			We have
			\begin{align}
				A_n(z)=&\frac{1}{z} + \sum_{\ell=2}^{m+1}\left(\frac{1}{z^\ell}\cdot\sum_{i=1}^{m+2-\ell}\frac{(\ell-2+i)t_{\ell-2+i}}{it_i}R_{n,i}\right),\label{An-m}\\
				 B_n(z)=&-\frac{n}{z}+\sum_{\ell=2}^{m+1}\left(\frac{1}{z^\ell}\cdot\sum_{i=1}^{m+2-\ell}\frac{(\ell-2+i)t_{\ell-2+i}}{it_i}r_{n,i}\right),\label{Bn-m}
			\end{align}
			where $\{R_{n,i},r_{n,i},i=1,2,\dots,m\}$ are defined by
			\begin{equation}\label{mdefRr}
				\begin{aligned}				 R_{n,i}(\vec{t}\:):=&\frac{it_i}{h_n}\:\int_{0}^{\infty}\frac{dy}{y^i}P_n^2(y;\vec{t}\:)w(y;\vec{t}\:),\\
					 r_{n,i}(\vec{t}\:):=&\frac{it_i}{h_{n-1}}\:\int_{0}^{\infty}\frac{dy }{y^i}P_n(y;\vec{t}\:)P_{n-1}(y;\vec{t}\:)w(y;\vec{t}\:).
				\end{aligned}
			\end{equation}
		\end{lemma}

		\subsection{Difference equations}
Plugging \eqref{An-m}-\eqref{Bn-m} into $(S_1)$ given in the Introduction, by equating the coefficients of $z^{-j}  (j=1,2,\dots, m+1)$ in the resulting equation, we obtain
\begin{align}
&z^{-1}:&\alpha_{n}=2n+1+\alpha+\sum_{i=1}^m R_{n,i},\label{mS1eq1}\\
&z^{-(m+1)}:&r_{n+1,1}+r_{n,1}+\al_nR_{n,1}-t_1=0,\label{mS1eq4}
\end{align}
and for $\ell=2,\dots,m$,
			\begin{equation}\label{mS1eq2}					 z^{-\ell}:\quad\sum_{i=1}^{m+2-\ell}\frac{(\ell-2+i)t_{\ell-2+i}}{it_i}(r_{n+1,i}+r_{n,i}+\al_n R_{n,i})					 =\sum_{i=1}^{m+1-\ell}\frac{(\ell-1+i)t_{\ell-1+i}}{it_i}R_{n,i}
					+(\ell-1)t_{\ell-1}.
				\end{equation}
For $\ell=m$ in \eqref{mS1eq2}, with the aid of \eqref{mS1eq4}, we get
			\begin{equation*}
	 r_{n+1,2}+r_{n,2}+\al_nR_{n,2}-\tau R_{n,1}=0,
			\end{equation*}
where $\tau=2t_2/t_1$. When $\ell=m-1$ in \eqref{mS1eq2}, in light of the above equation and \eqref{mS1eq4}, we find
\begin{equation*} r_{n+1,3}+r_{n,3}+\al_nR_{n,3}-\rho R_{n,2}=0,
			\end{equation*}
where $\rho=3t_3/(2t_2)$.
We proceed in this manner and finally obtain
			\begin{equation}\label{mS1eq5}
				 r_{n+1,j}+r_{n,j}+\al_nR_{n,j}-\frac{jt_j R_{n,j-1}}{(j-1)t_{j-1}}=0,
			\end{equation}
			for $j=2,\dots,m$.

Inserting \eqref{An-m}-\eqref{Bn-m} into $(S_2')$, by comparing the coefficients of $z^{-j} (j=2,3,\dots,2m+2)$ on its both sides, we are led to
				\begin{align}
&z^{-2}:&					 \beta_n =n(n+\alpha)+\sum_{i=1}^{m}\Big(\sum_{j=0}^{n-1}R_{j,i}+r_{n,i} \Big),
				\end{align}
				\begin{align*}
&z^{-3}:\qquad&			 \sum_{i=1}^{m-1}\frac{(1+i)t_{1+i}}{it_i}\Big(\sum_{j=0}^{n-1}R_{j,i}+r_{n,i} \Big)-(2n+\alpha)\sum_{i=1}^{m}r_{n,i}+nt_1=\beta_n\cdot \sum_{i=1}^{m}(R_{n-1,i}+R_{n,i}),
				\end{align*}
$z^{-\ell}$ for
$4 \leq \ell\leq m+1$:
\begin{equation}
\begin{aligned} &\sum_{k=2}^{\ell-2}\left(-(k-1)t_{k-1}+\sum_{i=1}^{m+2-k}\frac{(k-2+i)t_{k-2+i}}{it_i}r_{n,i} \right)\left(\sum_{i=1}^{m+2-\ell+k}\frac{(\ell-k-2+i)t_{\ell-k-2+i}}{it_i}r_{n,i} \right)\\					 &+\sum_{i=1}^{m+2-\ell}\frac{(\ell-2+i)t_{\ell-2+i}}{it_i}\Big(\sum_{j=0}^{n-1}R_{j,i}+r_{n,i} \Big)-(2n+\alpha)\sum_{i=1}^{m+3-\ell}\frac{(\ell-3+i)t_{\ell-3+i}}{it_i}r_{n,i}+n(\ell-2)t_{\ell-2}\\&=\beta_n \left[ \sum_{i=1}^{m+3-\ell}\frac{(\ell-3+i)t_{\ell-3+i}}{it_i}(R_{n-1,i}+R_{n,i})\right.\\
&\quad \left.\qquad+\sum_{k=2}^{\ell-2}\left( \sum_{i=1}^{m+2-k}\frac{(k-2+i)t_{k-2+i}}{it_i} R_{n,i}\right) \left(\sum_{i=1}^{m+2-\ell+k}\frac{(\ell-k-2+i)t_{\ell-k-2+i}}{it_i}  R_{n-1,i}\right)  \right],
\end{aligned}
\end{equation}
$z^{-(m+2)}:$
\begin{equation}\label{S2'eq4-m}					
\begin{aligned}					 &\sum_{k=2}^{m}\left(-(k-1)t_{k-1}+\sum_{i=1}^{m+2-k}\frac{(k-2+i)t_{k-2+i}}{it_i}r_{n,i} \right)\cdot \left(\sum_{i=1}^{k}\frac{(m-k+i)t_{m-k+i}}{it_i}r_{n,i} \right)\\
&\qquad-(2n+\alpha)\frac{mt_m}{t_1}r_{n,1}+n\cdot mt_m\\
&=\beta_n \left[\sum_{k=2}^{m}\left(\sum_{i=1}^{m+2-k}\frac{(k-2+i)t_{k-2+i}}{it_i}R_{n,i} \right) \cdot \left( \sum_{i=1}^{k}\frac{(m-k+i)t_{m-k+i}}{it_i}R_{n-1,i}\right)\right.\\
&\left.\qquad\qquad+\frac{mt_m}{t_1}   (R_{n-1,1}+R_{n,1}) \right],
\end{aligned}
\end{equation}
$z^{-\ell}$ for $m+3\leq \ell \leq2m+1$:
\begin{equation}\label{S2'eq5-m}
					\begin{aligned}						 &\sum_{k=\ell-(m+1)}^{m+1}\left(-(k-1)t_{k-1}+\sum_{i=1}^{m+2-k}\frac{(k-2+i)t_{k-2+i}}{it_i}r_{n,i} \right) \left(\sum_{i=1}^{m+2-\ell+k}\frac{(\ell-k-2+i)t_{\ell-k-2+i}}{it_i}r_{n,i} \right)\\&=\beta_n \sum_{k=\ell-(m+1)}^{m+1}\left(\sum_{i=1}^{m+2-k}\frac{(k-2+i)t_{k-2+i}}{it_i}R_{n,i} \right)\left(\sum_{i=1}^{m+2-\ell+k}\frac{(\ell-k-2+i)t_{\ell-k-2+i}}{it_i}R_{n-1,i} \right),
\end{aligned}
\end{equation}
and
\begin{align}\label{S2'eq6-m}
z^{-(2m+2)}:\qquad\qquad \qquad r_{n,1}(r_{n,1}-t_1)=\beta_nR_{n,1}R_{n-1,1}.
\end{align}

\begin{remark}\label{btRn-1}
Using the $m$ equations given by \eqref{S2'eq5-m}-\eqref{S2'eq6-m}, following a similar procedure to that for $m=3$, we express $\beta_nR_{n-1,i}$ by $\{R_{n,j},r_{n,j}, j=1,\dots,i\}$ for $i=1,\dots,m$. Plugging these $m$ expressions into \eqref{S2'eq4-m} gives $\beta_n$ in terms of $\{R_{n,j},r_{n,j}, j=1,\dots,m\}$.

In addition, by inserting \eqref{mS1eq1} into \eqref{mS1eq4} and \eqref{mS1eq5}, and combining the obtained equations with \eqref{S2'eq4-m}-\eqref{S2'eq6-m}, we arrive at a system of differences equations for $\{R_{n,j},r_{n,j}, j=1,\dots,m\}$, which can be iterated in $n$.
\end{remark}

		\subsection{Toda-like Equations and PDEs}
		Differentiating both sides of
		\begin{align*}		 h_n(\vec{t}\:)=\int_{0}^{+\infty}P_n^2(x;\vec{t}\:)w(x;\vec{t}\:)dx,\\	 0=\int_{0}^{+\infty}P_n(x;\vec{t}\:)P_{n-1}(x;\vec{t}\:)w(x;\vec{t}\:)dx
		\end{align*}
with respect to $t_i$ for $i=1,\dots,m$, we get
				\begin{align}
					it_i\partial_{t_i}\ln h_n=&-R_{n,i},\label{Dhp-m}\\
it_i\partial_{t_i}p(n,\vec{t}\:)=&r_{n,i}.\nonumber
				\end{align}
In view of $\beta_n=h_n/h_{n-1}$ and $\alpha_n=p(n,\vec{t}\:)-p(n+1,\vec{t}\:)$, it follows that
				\begin{align}
					it_i\partial_{t_i}\ln \beta_n= &R_{n-1,i}-R_{n,i},\label{Dbtal-m}\\
it_i\partial_{t_i}\alpha_n=&r_{n,i}-r_{n+1,i},\nonumber
				\end{align}
for $i=1,\dots,m$.

		\begin{proposition}
		The recurrence coefficients satisfy the following partial differential-difference equations
			\begin{align*}		 \sum_{i=1}^{m}it_i\partial_{t_i}\alpha_n=\beta_n-\beta_{n+1}+\alpha_n,\\
				\sum_{i=1}^{m}it_i\partial_{t_i}\ln \beta_n=\alpha_{n-1}-\alpha_n+2,
			\end{align*}
from which follow a second order Toda-like equation for $\beta_n:$
			\begin{equation*}
				\Big( \sum_{i=1}^{m}(it_i)^2\partial_{t_it_i}+\sum_{1\leq i<k\leq m}ikt_it_k(\partial_{t_it_k}+\partial_{t_kt_i})+\sum_{i=1}^{m}i(i-1) t_i\partial_{t_i} \Big) \ln \beta_n=\beta_{n-1}-2\beta_n+\beta_{n+1}-2.
			\end{equation*}
		\end{proposition}

Define
			\begin{equation*}
	 H_n(\vec{t}\:):=\sum_{i=1}^{m}it_i\partial_{t_i}\ln D_n(\vec{t}\:).
			\end{equation*}
According to $D_n=\prod_{j=0}^{n-1}h_j$ and \eqref{Dhp-m}, we have
			\begin{equation}\label{Hn-m}				 H_n(\vec{t}\:)=-\sum_{i=1}^{m}\sum_{j=0}^{n-1}R_{j,i}.
			\end{equation}
Since $p(n,\vec{t}\:)=-\sum_{j=0}^{n-1}\alpha_j$, it follows from \eqref{mS1eq1} that
			\begin{align} p(n,\vec{t}\:)&=-n(n+\alpha)-\sum_{j=0}^{n-1}\sum_{i=1}^{m}R_{j,i},\label{psumR-m}\\
			    &=\sum_{i=1}^{m}r_{n,i}-\beta_n.\label{psumr-m}
			\end{align}
A combination of \eqref{psumR-m} and \eqref{Hn-m} yields
			\begin{equation}\label{Hnpn-m}
				 H_n(\vec{t}\:)=p(n,\vec{t}\:)+n(n+\alpha),
			\end{equation}
which, on account of \eqref{Dhp-m}, results in
			\begin{equation}\label{Hnrn-m}
				 it_i\partial_{t_i}H_n(\vec{t}\:)=r_{n,i},
			\end{equation}
for $i=1,2,\dots,m$. Inserting them into \eqref{psumr-m}, in light of \eqref{Hnpn-m}, we are led to
			\begin{align}\label{btHn-m}
				 \beta_n &=\Big(\sum_{i=1}^mit_i\partial_{t_i}\Big)H_n-H_n+n(n+\alpha).
			\end{align}
 We have mentioned in Remark \ref{btRn-1} that $\beta_nR_{n-1,i} (i=1,\dots,m)$ can be expressed in terms of $\{R_{n,j},r_{n,j}, j=1,\dots,m\}$. Combining these expressions with \eqref{Dbtal-m}, we get $m$ equations involving the $2m$ auxiliary quantities, $\beta_n$ and $it_i\partial_{t_i}\beta_n$, for $i=1,\dots,m$. We solve $\{R_{n,j}, j=1,2,\dots,m\}$ from these equations, as was done for the $m=3$ case prior to Remark \ref{R1-3}, and substitute them into the expression for $\beta_n$ in terms of the auxiliary quantities (noted in Remark \ref{btRn-1}). Combining the resulting equation with \eqref{btHn-m} and replacing $r_{n,i}$ using \eqref{Hnrn-m}, we arrive at the second order PDE satisfied by $H_n$.

\section{Conclusions}
We have studied the partition function of the singularly perturbed Laguerre unitary ensemble associated with the weight function $x^{\alpha}\exp(-x-\sum_{k=1}^m t_k/x^k)$ with $\alpha>-1$ and $m\geq2$. The ladder operators for the monic orthogonal polynomials associated with Laguerre-type weight function $x^{\alpha}w_0(x)$ which decreases sufficiently fast at $+\infty$,  with $\alpha>-1$ and $w_0(x)>0$ being continuously differentiable, were deduced recently. By using them, we first provide a detailed investigation of the $m=2$ case to derive a second order PDE satisfied by the logarithmic derivative of the partition function, reducing to the $\sigma$-form of a Painlev\'{e} III$'$ equation when $t_2\rightarrow0^+$.
The derivation is not a straightforward generalization of the $m=1$ case which was studied in prior work and is actually much more involved: the expression of the recurrence coefficient $\beta_n$ in terms of the auxiliary quantities (see \eqref{betanRnrn}) and the determination of the sign before the square root in the expression for the auxiliary variable (see \eqref{Rnsol0}-\eqref{Rnsol}), which are crucial for the derivation, is somewhat tricky to obtain; although the PDE was quite complicated initially, we have simplified it  to a fairly compact form (see \eqref{PDEHn}).  In addition, under a suitable double scaling, we obtain the limiting form of the PDE. The derivation is non-trivial and requires careful analysis, especially on the deduction of the relations between the double scaled auxiliary quantities and the logarithmic derivative of the Hankel determinant (see Proposition \ref{RRstarH}).

An analogous discussion is then conducted for the case $m=3$, which serves as a template for the study of the general case. We have outlined the derivation procedure for general $m$. Although the final equation is not explicitly written due to its complexity, a PDE can, in principle, be established for arbitrary $m$.

The approach presented in this paper fails when one or more of the $t_i$ $(i=1,2,\dots,m-1)$ vanish and a new method is needed to address this case in the future. In addition,  we may follow the line of reasoning in this paper to study the Gaussian and Jacobi problems of the same type. However, the process is likely more involved since these two cases require the introduction of more auxiliary quantities.

\section*{Acknowledgements}
This work was supported by National Natural Science Foundation of China under grant numbers 12101343 and 12371257, and by Shandong Provincial Natural Science Foundation with project number ZR2021QA061.

\section*{Appendix: Some Relevant Integral Identities}

		We state the following integral identities for $0<a<b$, which are relevant to our derivation and could be found in \cite{ChenMcKay2012} and \cite{GradshteynRyzhik2007}:
		\begin{align*}
		\int_{a}^{b}\frac{dx}{\sqrtx}=&\pi,\\
		\int_{a}^{b}\frac{\ln x\, dx}{\sqrtx}=&2\pi\ln \frac{\sqrt{a}+\sqrt{b}}{2},\\
		\int_{a}^{b}\frac{x\, dx}{\sqrtx}=&\frac{a+b}{2}\pi,\\
		 \int_{a}^{b}\frac{dx}{x\sqrtx}=&\frac{\pi}{\sqrt{ab}},\\	 \int_{a}^{b}\frac{dx}{x^2\sqrtx}=&\frac{(a+b)}{2(ab)^{3/2}}\pi,\\	 \int_{a}^{b}\frac{dx}{x^3\sqrtx}=&\frac{3(a+b)^2-4ab}{8(ab)^{5/2}}\pi\label{app5}.
		\end{align*}

		\bibliographystyle{plain}
		}

\end{document}